\documentclass[11pt, letterpaper]{scrartcl}
\usepackage{mydef2col}
\usepackage[autostyle=true]{csquotes}
\usepackage{subfig}
\usepackage{colortbl}

\allowdisplaybreaks[3]

\usepackage{tikz}
\usetikzlibrary{calc,decorations.markings}
\usetikzlibrary{arrows,patterns,positioning}

\def\BS{{BS(t,S,\sigma)}}
\def\sig{{\mathfrak{S}}}
\def\Sig{{\Sigma}}
\def\tSig{{\tilde{\Sig}}}
\def\calL{{\mathcal L}}
\def\calM{{\mathcal M}}
\def\calE{{\mathcal E}}
\def\dT{{\Delta T}}
\def\pd{{\partial}}
\def\pade{Pad{\'e} }
\def\frechet{Fr{\`e}chet }
\def\gato{G{\^a}teaux }
\def\nl{{\mathrm{NL}}}

\def\sdt{\sqrt{\Delta T}}

\title{A model-free backward and forward nonlinear PDEs for implied volatility}
\author{
\authorstyle{Peter Carr\textsuperscript{1}},
\authorstyle{Andrey Itkin\textsuperscript{1}} and
\authorstyle{Sasha Stoikov\textsuperscript{2}}
\newline
\newline
\textsuperscript{1}
\institution{Tandon School of Engineering, New York University, 12 Metro Tech Center, 26th floor, Brooklyn NY 11201, USA}
\newline
\textsuperscript{2}
\institution{Cornell University, School of Operations Research and Information Engineering, 2W Loop Road,
New York, NY 10044, USA}
}

\date{\today}
\begin{document}

\maketitle

\lettrineabstract{We derive a backward and forward nonlinear PDEs that govern the implied volatility of a contingent claim whenever the latter is well-defined. This would include at least any contingent claim written on a positive stock
price whose payoff at a possibly random time is convex. We also discuss suitable initial and boundary conditions for those PDEs. Finally, we demonstrate how to solve them numerically by using an iterative finite-difference approach.
}

\vspace{0.5in}

\section*{Introduction}

As this is well-known, given an option contract written, e.g., on some stock, the implied volatility is derived from an option price and shows what the market implies about the underlying stock volatility in the future. For instance, the implied volatility is one of six inputs used in a simple option pricing (Black-Scholes) model, but is the only one that is not directly observable in the market. The standard way to determine it by knowing the market price of the contract and the other five parameters, is solving for the implied volatility by equating the model and market prices of the option contract.

There exist various reasons why traders prefer considering option positions in term of the implied volatility, rather than the option price itself, see e.g., \cite{natenberg1994option}. However, in this paper our goal is not to discuss the importance of this concept. Instead we focus on the way how the implied volatility is computed, and namely - on performance of this process.

As was mentioned, the traditional method of computing the Black-Scholes implied volatility is relatively simple. Perhaps, the main problem with such an approach occurs if one needs to simultaneously compute the implied volatility for a wide set of the model (contract) parameters. Then, the iterative root solver can converge slowly because i) there is no a universal good initial guess to kick off the iterations, and ii) for some regions of the model parameters the solution either doesn't exist at all, or is hard to be found as the sensitivity of the model price to changes in the implied volatility is very low.

Therefore, a surprising idea would be to replace an algebraic equation \eqref{def} used to determine the implied volatility with another object, for instance, a partial differential equation (PDE). Here the surprise means that the PDE is a more complicated object than the algebraic equation, and solving the PDE (especially nonlinear) requires some special methods usually much more complicated than a simple root solver. However, in contrast to the algebraic equation that should be solved independently at every given point e.g., in the space of option strikes and maturities, the PDE approach allows finding the solution simultaneously in many such the points by using a series of marching sweeps in time. Also, solving a linear PDEs  doesn't require iterations, and hence, there is no need for the good initial guess.

In this paper we demonstrate how such a PDE could be derived by using just a general definition of the implied volatility. Moreover, we further extend this idea by deriving several types of these PDEs. For instance, we obtain quasilinear backward and forward parabolic PDEs, and also nonlinear hyperbolic PDEs of the first order. As these PDEs are new and yet have not been investigated in the literature, we need to discuss suitable initial and boundary conditions that provide the solution of these PDE to exist and be unique. We also develop an iterative numerical method to solve the PDEs by using a finite-difference approach. The method is of second order of approximation in both space and time, is unconditionally stable and preserves positivity of the solution. Using this method we compute the PDE implied volaitlity and find that our intuition behind the main idea of the paper is correct. In other words, performance of the finite-difference solver exceeds that of the traditional approach by factor of forty. However, this result is subject to some details which are highlighted in the last Section.

Despite in this paper the PDE approach is considered as an alternative to the traditional method of getting the implied volatility by using the root solver, the latter in turn was seriously improved within the last few years. First, in \cite{Jackel2015} modified Newton iterations were proposed to solve \eqref{def}.  When doing so the entire input domain is decomposed into four areas. Rational approximations are used to provide initial guesses and reduce the number of iterations. Also this method provides an accuracy close to the machine precision.  However, it breaks possible vectorization (but we don't consider this topic in this paper anyway).

In \cite{Glau2018} the approach of \cite{Jackel2015}  is improved by using a bivariate Chebyshev interpolation of the implied volaitlity. This algorithm allows vectorization and consists of two steps. At the offline phase polynomial weights of a low-rank Chebyshev interpolation of the implied volatility surface are computed and stored for four different input domains, using the algorithm of \cite{Jackel2015}. This step is only performed once. During the online phase the input data is split into four domains, and the Chebyshev interpolation is applied to each domain by choosing the precomputed nodes from the offline step.

From this prospective, this problem can also be solved by building an artificial neural network (ANN) and training it on a given data set which will include all parameters  of the contract (the spot price, the time to maturity, the strike, option price, etc.) and output the Black-Scholes implied volatility. The method of \cite{Jackel2015} can be used when training this ANN as the golden "pricer". As such, the approach proposed in this paper becomes, apparently, less important since the ANN training could be done offline, while performance of the "pricer" is less important, rather than the accuracy. This is completely opposite to the proposed PDE approach. However, using the ANN still requires solving various problems, such as no-arbitrage, the existence of derivatives, etc., see \cite{ItkinANN}. Also, the advantage of using the ANN approach vs traditional methods, is, perhaps, a property inherent to many financial models, and not only to the model we consider in this paper.

Aside of practical importance, the existence of the model-free quasilinear parabolic PDE (or the other hyperbolic PDE) governing the implied volatility is an interesting fact, which yet has not been discovered in the literature. Therefore, the novelty of this paper consists in the derived equations together with the proposed numerical method of their solution.

The rest of the paper is organized as follows. In Section~\ref{backwardPDE} we derive a backward quasilinear PDE for the implied volatility. In Section~\ref{forwardEq} a similar approch is used to derive an analogous forward PDE. Section~\ref{bc} discusses how one can set the appropriate initial and boundary conditions for these PDEs. In Section~\ref{otherPDE} we also derive some  other PDEs for the implied volatility that have a different type (nonlinear hyperbolic rather than quasilinear parabolic equations).   Then we construct a numerical method to solve these equations which is described in detail in Section~\ref{solPDE} and two Appendices. The results of our numerical experiments and comparison of two methods: i) the traditional method of finding the implied volatility by using a root solver with the Black-Scholes equation, and ii) the proposed method of solving the derived PDEs numerically, are provided in Section~\ref{results}. In Section~\ref{discussion} we discussed the results obtained and conclude.

\section{Backward PDE} \label{backwardPDE}

Consider a class of contingent claims which have a single payoff and which are written on the spot price path of a single underlying risky asset. The Black-Scholes implied volatility, \cite{hull:97}, can be defined for any such contingent claim, as long as the claim's value in the Black-Scholes model is monotonic in the volatility of the underlying asset. Standard examples include European options and American options that are optimally held alive. If the claim's value is decreasing in volatility, then the value of a short position in that claim is increasing in volatility. Hence, there is no loss in generality if implied volatility is only defined for short or long positions in claims whose value increases with volatility. We refer to any such claim as an option. This section derives a new nonlinear partial differential equation (PDE) that relates the Black-Scholes implied volatility to the option price, to the underlying stock price, and to calendar time.

The Black-Scholes implied volatility is a useful measure, as it is a market practice instead  of quoting the option premium in the relevant currency, the options are quoted in terms of the Black-Scholes implied volatility. Over the years, option traders have developed an intuition in this quantity.  However, it can be further generalized by using a similar concept, but replacing the Black-Scholes framework with another one. For instance, in \cite{Schoutens2009} this is done under a \LY framework, and, therefore, based on distributions that match more closely historical returns. In this paper we don't consider these generalizations, and are concentrated only on the Black-Scholes implied volatility.

Consider a fixed time horizon $t \in [0,T]$ where $T > 0$ is the time to maturity. We assume that the underlying stock price $S$ is strictly positive over this horizon. Let $BS(S,K,\sigma,r,q,t)$ denote the function relating the Black-Scholes model value $BS(S,K,\sigma,r,q,t) \in \mathbb{R}$ of a contingent claim to the underlying stock price level $S > 0$, the strike $K > 0$, the (constant) volatility $\sigma > 0$, calendar time $t \in [0,T]$, the constant interest rate $r$ and continuous dividend $q$. Let subscripts denote partial derivatives. We assume that the function $BS(S,K,\sigma,r,q,t)$ is governed by the Black-Scholes backward PDE. To make the notation easier, further in the expression $BS(S,K,\sigma,r,q,t)$ we will drop all variables not relevant to our particular context. Therefore, instead, in this Section we denote the function value as $BS(S,\sigma,t)$. With this notation, the corresponding Black-Scholes PDE reads, \cite{hull:97}
\begin{equation}
\fp{\BS}{t} + (r-q) S \fp{\BS}{S} + \frac{1}{2}\sigma^2 S^2 \sop{\BS}{S} = r \BS,
\label{pde}
\end{equation}
\noindent defined on some open region ${\cal R}$ of space and time. For example, for a European Call or Put, the region ${\cal R}$ would be the Cartesian product ${\cal R}_E \equiv [0,\infty) \times [0,T)$. As a second example, for a down-and-out Call with a lower barrier $L \in (0,S_0)$, the region ${\cal R}$ would be the Cartesian product ${\cal R}_L \equiv (L,\infty) \times [0,T)$.

We furthermore assume that the Vega of the claim is always strictly positive, i.e. $BS_{\sigma}(S,\sigma,t)>0$ on ${\cal R}$. This clearly holds for a European Call or Put on ${\cal R}_E$, but it can also be shown that it holds for a down-and-out Call on ${\cal R}_L$. We also assume that one can directly observe the market price $Z(T,K)$ of such a claim, and that this market price lies in some arbitrage-free open interval ${\cal A} \subset \mathbb{R}$. For a European Call struck at $K$, ${\cal A}_C \equiv ((Q_T S - D_T K)^+, S Q_T)$, where $D_T = e^{-r T}$ is the discount factor, and $Q_T = e^{-q T}$ is the price appreciation factor. For a European Put struck at $K$, ${\cal A}_P \equiv ((D_T K - Q_T S)^+, D_T K)$. For a down-and-out Call struck at $K$, ${\cal A}_D \equiv ((D_T K - Q_T S)^+, Q_T(S-L))$, \cite{hull:97}.

For any such claim, we can implicitly define a function $\Sigma(S,t,K,T,r,q,Z)$ that relates the implied volatility $\Sigma$ of the contingent claim to $(S,t) \in {\cal R}$ and to the market price of the contingent claim $Z \in {\cal A}$. Again, for easy of notation we drop $K,T, r,q$ from the list of independent variables, and then the definition of $\Sigma(S,t,Z)$ reads
\begin{equation} \label{def}
BS(S,\Sigma(S,t,Z),t)  = Z.
\end{equation}
Let ${\cal D}$ denote the domain for this function, which is defined as the Cartesian product of ${\cal R}$ and ${\cal A}$. Since the PDE \eqref{pde} holds for any level of $\sigma >0$, it holds in particular if we set $\sigma=\Sigma(S,t,Z)$:
\begin{equation}
\Bigg\{ \fp{\BS}{t} + (r-q) S \fp{\BS}{S} + \frac{1}{2}\sigma^2 S^2 \sop{\BS}{S} - r \BS \Bigg\}
\biggr|_{\sigma=\Sigma(S,t,Z)} = 0.
\label{pde1}
\end{equation}
We now show that \eqref{def} and \eqref{pde1} can be used to generate a nonlinear PDE governing $\Sigma(S,t,Z)$ on ${\cal D}$.

First, differentiate \eqref{def} with respect to $t$:
\begin{equation}
BS_t(S,\Sigma(S,t,Z),t) + BS_{\sigma}(S,\Sigma(S,t,Z),t) \Sigma_t(S,t,Z) = 0.
\label{sigt}
\end{equation}

Second, differentiate \eqref{def} with respect to $S$ instead:
\begin{equation}
BS_S(S,\Sigma(S,t,Z),t) + BS_{\sigma}(S,\Sigma(S,t,Z),t) \Sigma_S(S,t,Z) = 0.
\label{sigs}
\end{equation}

In what follows, we will also drop the 3 arguments of $\Sigma$ for notational ease, as it shouldn't give rise to any confusion. Now differentiate \eqref{sigs} with respect to $S$:
\begin{equation}
BS_{SS}(S,\Sigma,t) + 2 BS_{S\sigma}(S,\Sigma,t) \Sigma_S + BS_{\sigma\sigma}(S,\Sigma,t) \Sigma^2_S + BS_{\sigma}(S,\Sigma,t) \Sigma_{SS} = 0.
\label{sigss}
\end{equation}

Substituting \eqref{sigt} and \eqref{sigss} into \eqref{pde1} implies:
\begin{align} \label{pde2}
- BS_{\sigma}(S,\Sigma,t)\Sigma_{t} &- \frac{1}{2}\Sigma^2 S^2 \left[ 2 BS_{S\sigma}(S,\Sigma,t) \Sigma_S + BS_{\sigma\sigma}(S,\Sigma,t) \Sigma^2_S + BS_{\sigma}(S,\Sigma,t) \Sigma_{SS}\right] \\
&+ (r-q) S BS_{S}(S,\Sigma,t) = r BS(S,\Sigma,t). \nonumber
\end{align}
We now show that we can express Vanna $BS_{S\sigma}(S,\Sigma,t)$ and Volga $BS_{\sigma \sigma}(S,\Sigma,t)$ in terms of Vega $BS_{\sigma}(S,\Sigma,t)$. First, differentiate \eqref{def} with respect to $Z$:
\begin{equation} \label{sigz}
BS_{\sigma}(S,\Sigma,t) \Sigma_Z = 1.
\end{equation}
Next, differentiate \eqref{sigz} with respect to $S$:
\begin{equation}
BS_{S\sigma}(S,\Sigma,t) \Sigma_Z + BS_{\sigma\sigma}(S,\Sigma,t) \Sigma_S \Sigma_Z + BS_{\sigma}(S,\Sigma,t). \Sigma_{SZ} = 0.
\label{sigzs}
\end{equation}
Now differentiate \eqref{sigz} with respect to $Z$ instead:
\begin{equation}
BS_{\sigma\sigma}(S,\Sigma,t) \Sigma^2_Z + BS_{\sigma}(S,\Sigma,t) \Sigma_{ZZ} = 0.
\label{sigzz}
\end{equation}
Solving \eqref{sigzz} for Volga $BS_{\sigma \sigma}(S,\Sigma,t)$ yields
\begin{equation}
BS_{\sigma\sigma}(S,\Sigma,t) = -  BS_{\sigma}(S,\Sigma,t) \frac{\Sigma_{ZZ}}{\Sigma^2_{Z}}.
\label{sol}
\end{equation}
Substituting \eqref{sol} into \eqref{sigzs} and solving for Vanna $BS_{S\sigma}(S,\Sigma,t)$ implies:
\begin{equation}
BS_{S\sigma}(S,\Sigma,t) = BS_{\sigma}(S,\Sigma,t) \frac{\Sigma_{S}}{\Sigma^2_{Z}}\Sigma_{ZZ} - BS_{\sigma}(S,\Sigma,t) \frac{\Sigma_{SZ}}{\Sigma_{Z}}.
\label{sol1}
\end{equation}
Substituting \eqref{sol} and \eqref{sol1} into \eqref{pde2} yields
\begin{align} \label{pde4}
0 &= BS_{\sigma}(S,\Sigma,t)  \Bigg\{
\Sigma_{t} + \frac{1}{2} \Sigma^2 S^2 \left[ \left(\frac{\Sigma_S}{\Sigma_Z} \right)^2 \Sigma_{ZZ} - 2 \frac{\Sigma_ S}{\Sigma_Z}  \Sigma_{SZ} + \Sigma_{SS}\right]  \Bigg\} \\
&- (r-q) S BS_{S}(S,\Sigma,t) + r Z. \nonumber
\end{align}
Also from \eqref{sigz} we can express the option Vega via $\Sigma_Z$, from \eqref{sigs} - the option Delta, and substitute these expressions into \eqref{pde4} to obtain
\begin{align} \label{pde4-2}
0 &= BS_{\sigma}(S,\Sigma,t)  \Bigg\{
\Sigma_{t} + \frac{1}{2} \Sigma^2 S^2 \left[ \left(\frac{\Sigma_S}{\Sigma_Z} \right)^2 \Sigma_{ZZ} - 2 \frac{\Sigma_ S}{\Sigma_Z}  \Sigma_{SZ} + \Sigma_{SS}\right]
+ (r-q) S \Sigma_S + r Z \Sigma_Z  \bigg\}.
\end{align}

Therefore, if the option Vega is non-zero, the backward PDE for the implied volatility $\Sigma$ finally takes the form
\begin{align} \label{pde4r}
\Sigma_{t} &+ \frac{1}{2} \Sigma^2 S^2 \left[ \left(\frac{\Sigma_S}{\Sigma_Z} \right)^2 \Sigma_{ZZ} - 2 \frac{\Sigma_ S}{\Sigma_Z}  \Sigma_{SZ} + \Sigma_{SS}\right]
+ (r-q) S \Sigma_{S} + r Z \Sigma_Z = 0.
\end{align}
This is a quasilinear\footnote{A partial differential equation is said to be quasilinear if it is linear with respect to all the highest order derivatives of the unknown function, \cite{Pinchover2005}.} PDE of the parabolic type.

When the Black-Scholes model Delta of a contingent claim is known in closed form, there is an alternative quasilinear PDE that governs the function $\Sigma(S,Z,t)$ on ${\cal D}$. Solving \eqref{sigz} for Vega $BS_{\sigma}(S,\Sigma,t)$, substituting  the result into \eqref{sigs}, and then solving for Delta $BS_S(S,\Sigma,t)$ implies:
 \begin{equation} \label{delta}
 BS_S(S,\Sigma,t) = - \frac{\Sigma_S}{\Sigma_Z}.
  \end{equation}

 Hence \eqref{pde4} can also be written as:
  \begin{equation} \label{pde4a}
  \Sigma_{t} + \frac{1}{2} \Sigma^2 S^2 \left[ BS^2_S(S,\Sigma,t) \Sigma_{ZZ} + 2 BS_S(S,\Sigma,t) \Sigma_{SZ} + \Sigma_{SS}\right] + (r-q) S \Sigma_{S} + r Z \Sigma_Z = 0.
  \end{equation}
To illustrate, it is well known that for a European Call, Delta $BS_S(S,\Sigma,t)$ is known in closed form:
 \begin{equation} \label{Nd1}
 BS_S(S,\Sigma,t) = N(d_1(S,\Sigma,t)),
  \end{equation}
 \noindent where:
 \begin{equation} \label{calldelta}
 d_1(S,\Sigma,t) \equiv \frac{\ln\frac{S}{K} + \frac{\Sigma^2}{2} T}{\Sigma \sqrt{T}},
 \end{equation}
 \noindent and $N(x)$ is the normal CDF.
 Hence, for a European Call, the quasilinear PDE \eqref{pde4a} can also be represented as another quasilinear PDE:
 \begin{equation}
 \Sigma_{t} + \frac{1}{2} \Sigma^2 S^2 \left[ {\cal N}^2(d_1(S,\Sigma,t)) \Sigma_{ZZ} + 2 {\cal N} (d_1(S,\Sigma,t))  \Sigma^2_{SZ} + \Sigma_{SS}\right] + (r-q) S \Sigma_{S} + r Z \Sigma_Z = 0,
 \label{pde5}
 \end{equation}
 \noindent whose domain is the Cartesian product of ${\cal R}_E$ and ${\cal A}_C$, which is
 $\mathcal{D} = (0,\infty) \times (0,T) \times ((Q_T S - D_T K)^+, Q_T S).$ Similar quasilinear PDE's can be derived for European Puts or down-and-out Calls.

For an American put option with a positive early exercise premium, in the continuation region the relevant PDE is \eqref{pde4r}, since neither the value nor  the Delta is known in closed form.

\section{Forward equation} \label{forwardEq}

In practice it is often necessary to simultaneously compute the Black-Scholes implied volatility of many options written on the same underlying, but having different strikes and maturity. This is to construct so-called the implied volatility surface, given market quotes $Z(K,T)$ and the market values of $S,r,q$. However, it is well-known that solving this problem by using the backward PDE is time-consuming since for every pair $K,T$ a separate backward PDE has to be solved. In contrast, the forward equation can be used to efficiently do this in one sweep, see e.g., \cite{Dupire:94, Gatheral2006}.

Having this in mind, in this Section we derive  a forward PDE for the implied volatility $\Sigma(S,T,K, r,q,Z)$. For easy of notation we will write it as $\Sigma(T,K,Z)$, thus dropping parameters that don't change.  The derivation can be done in a way similar to that in Section~\ref{backwardPDE}. Therefore, here we omit a detailed algebra, and provide just short explanations.

Again, we can start with the definition of the implied volatility $\Sigma(K,T,Z)$ in \eqref{def}, which in new variables read
\begin{equation} \label{defF}
BS(K,\Sigma(K,T,Z),T)  = Z,
\end{equation}
\noindent where $\Sigma(K,T,Z) \in [0,\infty)$, and the domain of definition for $(K,T,Z)$ is ${\cal R}_{F,E} \equiv [0,\infty) \times (0,\infty)\times \mathcal{A}_C$.

To proceed, we need a forward PDE for the European option price $BS(K,T, \sigma) $ which is also known as Dupire's equation, \cite{Dupire:94}
\begin{align} \label{Dup}
\fp{BS(K,T, \sigma) }{T} &=   \frac{1}{2} \sigma^2(K) K^2 \sop{BS(K,T, \sigma) }{K}  - (r-q) K \fp{BS(K,T, \sigma) }{K} - q BS(K,T, \sigma).
\end{align}
 Since the PDE \eqref{Dup} holds for any level of $\sigma >0$, it holds, in particular, if we set $\sigma=\Sigma(K,T,Z)$.
 In what follows, we will drop the 3 arguments of $\Sigma$ for notational ease.

Differentiating \eqref{defF} with respect to $T$, $K$ and twice with respect to $K$ yields
\begin{align} \label{sigtF}
0 &= BS_T(K,\Sigma,T) + BS_{\sigma}(K,\Sigma,T) \Sigma_T, \\
0 &= BS_K(K,\Sigma,T) + BS_{\sigma}(K,\Sigma,T) \Sigma_K, \nonumber \\
0 & = BS_{KK}(K,\Sigma,T) + 2 BS_{K\sigma}(K,\Sigma,T) \Sigma_K + BS_{\sigma\sigma}(K,\Sigma,T) \Sigma^2_K + BS_{\sigma}(K,\Sigma,T) \Sigma_{KK}. \nonumber
\end{align}

Substituting \eqref{sigtF} into \eqref{Dup} implies:
\begin{align} \label{pde2F}
0 &= BS_{\sigma}(K,\Sigma,T)\Sigma_{T} - \frac{1}{2}\Sigma^2 K^2 \left[ 2 BS_{K\sigma}(K,\Sigma,T) \Sigma_K + BS_{\sigma\sigma}(K,\Sigma,T) \Sigma^2_K + BS_{\sigma}(K,\Sigma,T) \Sigma_{KK}\right]  \nonumber \\
&+ (r-q) K BS_{\sigma}(K,\Sigma,T) \Sigma_K - q BS(K,\Sigma,T).
\end{align}

Again, derivatives $BS_{K\sigma}(K,\Sigma,T)$ and $BS_{\sigma \sigma}(K,\Sigma,T)$ can be expressed in terms of $BS_{\sigma}(K,\Sigma,T)$. Differentiating \eqref{defF}, first with respect to $Z$, and then with respect either to $K$, or to $Z$, yields
\begin{align} \label{sigzF}
1 & = BS_{\sigma}(K,\Sigma,T) \Sigma_Z, \\
0 & = BS_{K\sigma}(K,\Sigma,T) \Sigma_Z + BS_{\sigma\sigma}(K,\Sigma,T) \Sigma_K \Sigma_Z + BS_{\sigma}(K,\Sigma,T) \Sigma_{KZ}, \nonumber \\
0 &= BS_{\sigma\sigma}(K,\Sigma,T) \Sigma^2_Z + BS_{\sigma}(K,\Sigma,T) \Sigma_{ZZ}. \nonumber
\end{align}
The last equation in this system can be solved for $BS_{\sigma\sigma}(K,\Sigma,T) $, and then the second one - for
$BS_{K\sigma}(K,\Sigma,T)$.  Substituting these solutions into \eqref{pde2F} yields
\begin{align} \label{pde4F}
0 &= BS_{\sigma}(K,\Sigma,T)  \Bigg\{
\Sigma_{T} - \frac{1}{2} \Sigma^2 K^2 \left[ \left(\frac{\Sigma_K}{\Sigma_Z} \right)^2 \Sigma_{ZZ} - 2 \frac{\Sigma_ K}{\Sigma_Z}  \Sigma_{KZ} + \Sigma_{KK}\right]
+ (r-q) K \Sigma_K \Bigg\} - q Z.
\end{align}
Using the first line in \eqref{sigzF}, this could be re-written as
\begin{align} \label{pde4F-2}
0 &= BS_{\sigma}(K,\Sigma,T)  \Bigg\{
\Sigma_{T} - \frac{1}{2} \Sigma^2 K^2 \left[ \left(\frac{\Sigma_K}{\Sigma_Z} \right)^2 \Sigma_{ZZ} - 2 \frac{\Sigma_ K}{\Sigma_Z}  \Sigma_{KZ} + \Sigma_{KK}\right]
+ (r-q) K \Sigma_K  - q Z \Sigma_z \bigg\}.
\end{align}
If the option Vega is non-zero, we finally obtain the forward PDE for the implied volatility $\Sigma(K,T,Z)$
\begin{align} \label{pde4Fin}
0 &= \Sigma_{T} - \frac{1}{2} \Sigma^2 K^2 \left[ \left(\frac{\Sigma_K}{\Sigma_Z} \right)^2 \Sigma_{ZZ} - 2 \frac{\Sigma_ K}{\Sigma_Z}  \Sigma_{KZ} + \Sigma_{KK}\right]
+ (r-q) K \Sigma_K - q Z \Sigma_Z.
\end{align}
This PDE has to be solved subject to some initial and boundary conditions which are discussed in Section~\ref{bc}.

Since we consider the European options,  the Black-Scholes model strike delta of a contingent claim is known in closed form. Then, similar to the previous Section, the PDE in \eqref{pde4Fin} could be slightly simplified. Indeed, as follows from the second line of \eqref{sigtF} and first line of \eqref{sigzF}
\begin{equation} \label{deltaF}
BS_K(S,\Sigma,T) = - \frac{\Sigma_K}{\Sigma_Z}.
\end{equation}
Hence \eqref{pde4Fin} can also be re-written as:
\begin{equation}
0 = \Sigma_{T} - \frac{1}{2} \Sigma^2 K^2 \left[ BS_K^2(K, \Sigma,T) \Sigma_{ZZ} + 2
BS_K(K, \Sigma,T)\Sigma_{KZ} + \Sigma_{KK}\right] + (r-q) K \Sigma_K - q Z \Sigma_Z.
\label{pde4aF}
\end{equation}
For instance, for a European Call, the strike delta $BS_K(K,\Sigma,T)$ can be found in closed form. Indeed, since the Black-Scholes European Call option price is a homogeneous function of order 1 in $(S,K)$ it is easy to find that
\begin{equation} \label{deltaK}
K BS_K(K,\sigma,T) = BS(K,\sigma,T) - S BS_S(K,\sigma,T) = - K e^{-r T} N(d_2(K,\sigma,T)),
\end{equation}
\noindent where $d_2 = d_1 - \sigma \sqrt{T}$.

\section{Initial and boundary conditions} \label{bc}

The non-linear PDE \eqref{pde4Fin} describes evolution of the function $\Sigma(K,T,Z) \in [0,\infty)$  at the domain $\mathcal{R}_{F,E}$, and has to be solved subject to some initial and boundary conditions. However, this immediately brings some problems.

Indeed, the initial condition for \eqref{pde4Fin} has to be set at $T=0$. From the Black-Scholes formula we know that the option Vega $BS_{\sigma}(K,\Sigma,T)$ vanishes at $T=0$. Therefore, the expression in curly braces in \eqref{pde4F} could have any value. Also, the equation $BS_{\sigma}(K,\Sigma,T) = 0$ doesn't have a solution for $\Sigma$. This is because at $T=0$ function $BS(K,\Sigma,0)$ gets its intrinsic value (e.g. for the Call option -  $BS(K,\Sigma,0) = (S-K)^+$) which doesn't depend on $\Sigma$. Accordingly, for the ITM options there is no $\Sigma$ which would make the intrinsic value zero, and for the OTM options any $\Sigma$ could be chosen.

To resolve this, we make a change of the dependent variable $\Sigma(K,T,Z) \mapsto \sig(K,T,Z) = \Sigma(K,T,Z)\sqrt{T}$. This pursues two goals. First, at $T=0$ a natural initial condition implies $\sig(K,0,Z) = 0$. Second, if we know $\sig(K,T,Z)$ given some values of $(K,T,Z)$, the value of $\Sigma(K,T,Z)$ can be easily restored.

The forward PDE for the new dependent variable $\sig(K,T,Z)$ can be derived from \eqref{pde4Fin}, which yields
\begin{align} \label{pdeNew}
\sig_{T} &= \frac{1}{2 T} \sig^2 K^2 \left[ \left(\frac{\sig_K}{\sig_Z} \right)^2 \sig_{ZZ} - 2 \frac{\sig_ K}{\sig_Z}  \sig_{KZ} + \sig_{KK}\right] - (r-q) K \sig_K + q Z \sig_Z + \frac{1}{2 T}\sig.
\end{align}
Since \eqref{pdeNew} is a 2D parabolic PDE (quasilinear), the boundary conditions must be set $\forall T > 0$.

Suppose that we consider an European Call defined at the domain $\Omega: [K,Z] \in [0,\infty]\times ((Q_T S - D_T K)^+, Q_T S)$ which is depicted in Fig.~\ref{domainOm}. A special attention should be drawn to two of these boundaries: $Z=0$ and $Z = Q_T S - D_t K$ - to decide whether the PDE needs any boundary condition at them.
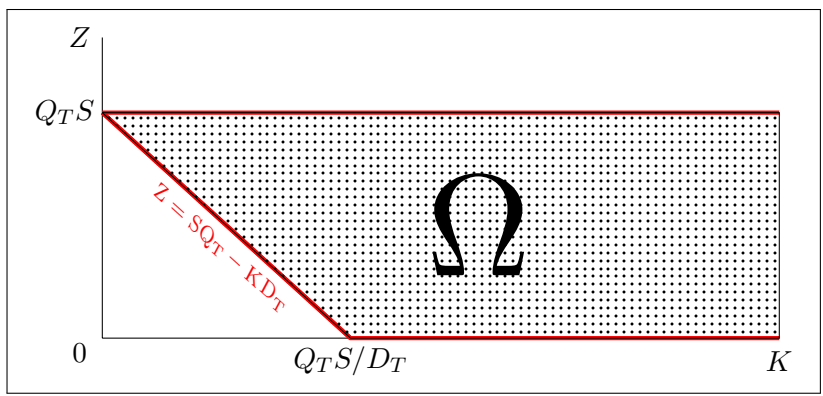
\begin{figure}[!htb]
\captionsetup{format=plain}
\begin{center}
\fbox{
\begin{tikzpicture}
\def\sizeV{4.};
\def\sizeH{9.};
\def\lA{2.5};
\def\lB{3.};
\def\lBB{3.3};
\def\lAB{2.};
\def\lBA{1.5};
\def\sh{0.03};
\def\sH{0.03};
\def\dis{1.5};

\draw (0,0) -- (\sizeH,0)
      (0,0) -- (0,\sizeV);
\draw[red, ultra thick] (0,\lB) -- (\lBB,0) -- (\sizeH,0);
\draw[red, ultra thick] (0,\lB) -- (\sizeH,\lB);
\draw[red, ultra thick] (0,\lB) -- (\lBB,0) node[midway,sloped,below,xslant=-0.5] {\scriptsize $Z = S Q_T - K D_T$};

\node at (\sizeH,-0.3){$K$};
\node at (-0.3,\sizeV){$Z$};
\node at (-0.3,-0.2){$0$};
\node at (-0.5,\lB){$Q_T S$};
\node at (\lBB,-0.3){$Q_T S/D_T$};
\node[scale=5] at (\sizeH/2+0.5,\lB/2){$\Omega$};

\draw[pattern=dots] (0,\lB) -- (\sizeH,\lB) -- (\sizeH,0) -- (\lBB,0) -- (0, \lB);
\end{tikzpicture}
}
\caption{The domain $\Omega:  [K,Z] \in [0,\infty]\times \mathcal{A}_C$.}
\label{domainOm}
\end{center}
\end{figure}

As mentioned in \cite{ItkinCarrBarrierR3}, the correct boundary conditions are determined by the speed of the diffusion term as we approach the boundary in a direction normal to the boundary. To illustrate, consider a PDE
\begin{equation} \label{oleinik}
C_t = a(x)C_{xx} + b(x)C_x + c(x) C,
\end{equation}
\noindent where $C= C (t,x)$ is some function of the time $t$ and the independent variable $x \in [0,\infty)$, $a(x), b(x), c(x) \in C^{2}$ are some known functions  of $x$. Then, as shown by \cite{OleinikRadkevich73}, no boundary condition is required at $x = 0$ if $\lim_{x \rightarrow 0}[b(x) - a_x(x)] \ge 0$. In other words, the convection term
at $x=0$ is flowing upwards and dominates as compared with  the diffusion term. A well-known example of such consideration is the Feller condition as applied to the Heston model, see, e.g., \cite{Lucic2008}.

To make it clear, no boundary condition means that instead of the boundary condition at $x \rightarrow 0$ the
PDE itself should be used at this boundary with coefficients $a(0), b(0), c(0)$.

In Section~\ref{solPDE} we show how to solve \eqref{pdeNew} by using temporal discretization that gives rise to the following equation
\begin{align} \label{pdeNewDis}
\sig_{T} &= \frac{1}{2 T} \tilde{\sig}^2 K^2 \left[ \left(\frac{\tilde{\sig}_K}{\tilde{\sig}_Z} \right)^2 \sig_{ZZ} - 2 \frac{\tilde{\sig_ K}}{\tilde{\sig}_Z}  \sig_{KZ} + \sig_{KK}\right] - (r-q) K \sig_K + q Z \sig_Z + \frac{1}{2 T}\sig.
\end{align}
Here $\sig = \sig(T+\Delta T, K,Z)$, $\tilde{\sig} = \tilde\sig(T, K,Z)$, and $\Delta T$ is the step of the temporal grid in $T$.
The \eqref{pdeNewDis} is a linear PDE since all coefficients are already known (they are either given, e.g., $r,q, T$, or
known from the previous time step). Thus, \eqref{pdeNewDis} is a 2D version of \eqref{oleinik}.

Accordingly, in the $K$ direction  with allowance for \eqref{sigtF},\eqref{deltaK}, the condition $\lim_{K \to 0} b(K) - a'(K) \ge 0$ reads
\begin{align*}
\lim_{K \to 0} & \left[b(K) - a'(K) \right] = \lim_{K \to 0} \left\{- \left[ (r-q)K +  \frac{\tilde{\sig} K}{T} \left( \tilde{\sig} + K \tilde{\sig}_K \right) \right] \right\} \\
&= \lim_{K \to 0} \left\{- K \left[ r-q +  \frac{\tilde{\sig}}{T} \left(\tilde{\sig} -
 K \sqrt{T} \frac{BS_K(K,\Sigma,T)}{BS_{\sigma}(K,\Sigma,T)}  \right) \right] \right\} \nonumber \\
&=  \lim_{K \to 0} \left\{- K \left[ r-q +  \frac{\tilde{\sig}}{T} \left(\tilde{\sig} +  \frac{N(d_2(K,\sigma,T))}{\Phi(d_2(K,\Sigma,T))}\right) \right] \right\} = 0, \qquad
\Phi(x) = N'(x), \nonumber
\end{align*}
\noindent because $\lim_{K \to 0} \tilde{\sig}(K,T,Z) = 0, \ T > 0$. Therefore, as per \cite{OleinikRadkevich73}, no boundary condition should be set at $K=0$ as the PDE itself serves as the correct boundary condition.

However, our domain $\Omega$ includes the only point with $K=0$, while the left boundary is a line $Z = Q_T S - D_T K$. It can be easily checked that along this line $\sig = 0$ which is the required boundary conditions.

In the $Z$ direction, again using \eqref{sigzF}, we obtain
\begin{align}
\frac{1}{2} K^2 & \fp{}{Z}\left( \frac{\tilde{\sig} \tilde{\sig}_K}{\tilde{\sig}_Z}\right)^2 = K^2 \frac{\tilde{\sig} \tilde{\sig}^2_K}{\tilde{\sig}_Z}\left( 1 + \frac{\tilde{\sig} \tilde{\sig}_{KZ}}{\tilde{\sig}_Z \tilde{\sig}_K} - \frac{\tilde{\sig \sig_{ZZ}}}{\tilde{\sig}^2_Z}.
\right) \\
&= K^2 T \frac{\tSig \tSig^2_K}{\tSig_Z}\left( 1 + \frac{\tSig \tSig_{KZ}}{\tSig_Z \tSig_K} - \frac{\tSig \tSig_{ZZ}}{\tSig^2_Z} \right) =
K^2 T \frac{\tilde{\Sig} \tilde{\Sig}^2_K}{\tilde{\Sig}_Z}\left( 1 + \tSig \frac{BS_{K,\sigma}(K,\tSig,T)}{BS_K(K,\tSig,T)} \right) \nonumber \\
&= K^2 T \frac{\tilde{\Sig} \tilde{\Sig}^2_K}{\tilde{\Sig}_Z}\left( 1 - \tSig  \frac{BS_{\sigma,K}(K,\tSig,T)}{D_T N(d_2(K,\tSig,T))} \right) =
K^2 T \frac{\tilde{\Sig} \tilde{\Sig}^2_K}{\tilde{\Sig}_Z}\left( 1 - \tSig  \frac{\Phi(d_2(K,\tSig,T)) d_1(K,\tSig,T)}{D_T N(d_2(K,\tSig,T))} \right). \nonumber
\end{align}
From \eqref{sigzF}, $\Sig_Z = 1/BS_\sigma(K,\tSig,T) > 0$. Therefore,
\begin{align*}
\lim_{Z \to 0} & \left[b(Z) - a'(Z) \right] = \lim_{Z \to 0} \left[q Z - K^2 \frac{\tilde{\Sig} \tilde{\Sig}^2_K}{\tilde{\Sig}_Z}\left(1 - \tSig  \frac{\Phi(d_2(K,\tSig,T)) d_1(K,\tSig,T)}{D_T N(d_2(K,\tSig,T))} \right) \right] \\
&= - K^2  \lim_{Z \to 0} \left[\tSig \tilde{\Sig}^2_K BS_\sigma(K,\tSig,T) \left(1 - \tSig  \frac{\Phi(d_2(K,\tSig,T)) d_1(K,\tSig,T)}{D_T N(d_2(K,\tSig,T))}\right) \right]
\nonumber
\end{align*}
Since at $Z=0$ the domain $\Omega$ is a straight line $K D_T \ge S Q_T$, it can be checked that the value of $d_1(K,\tSig,T)$ along this line is always negative, while $\tSig \ne 0$ except the point $K = S Q_T/D_T$. Therefore, what remains to be checked is the behavior of the option Vega as $Z \to 0$.
The implied volatility $\tSig$ which solves the equation $Z=0$ can be found numerically, and is not zero at $K D_T > S Q_T$. Therefore, Vega is positive (despite very small). Thus, $\lim_{Z \to 0} \left[b(Z) - a'(Z) \right] < 0$.
Therefore, as per \cite{OleinikRadkevich73}, we need a boundary condition at the boundary $Z=0$. This condition can be found by solving the equation
\begin{equation} \label{u0}
Z = BS(K,\Sig,T)=0,
\end{equation}
\noindent with respect to $\Sigma$. A typical example of the behavior of this solution is given in Fig.~\ref{solZ0} which is computed using the following input data: $r = 0.02, q = 0.01$ and two cases with $(S = 1.0, T = 0.1)$ and $(S=100, T=1.0)$.

\begin{figure}[!htb]
\begin{center}
\fbox{\includegraphics[width = 0.7\textwidth]{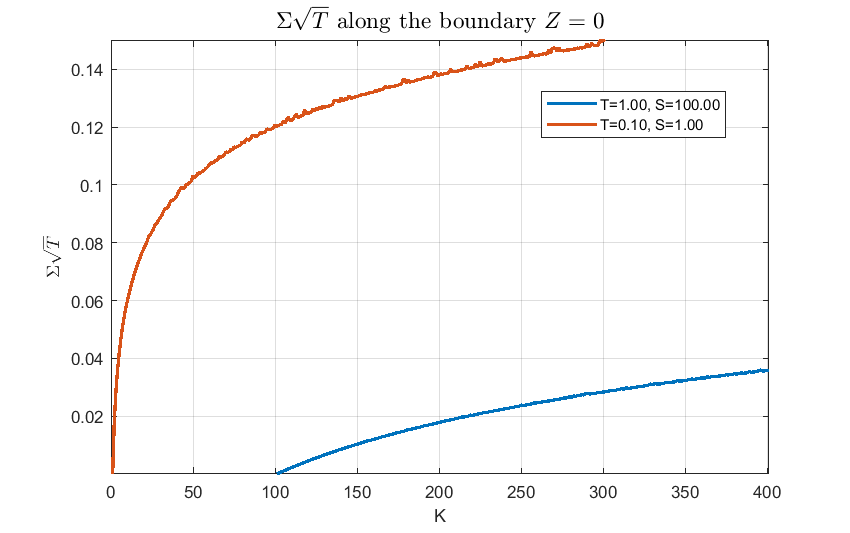}}
\caption{Implied volatility $\sig$ obtained by solving the equation $Z = 0$ for the European Call option.}
\label{solZ0}
\end{center}
\end{figure}

\paragraph{Boundary conditions at $K \to \infty$:}
To set this boundary condition we need to know an asymptotic behavior of the implied volatility when $K \to \infty$. This was a subject of an intensive research for the last decade, and several useful results are available in the literature, see, e.g., \cite{Lee2004,BenaimFritz2009, Archil2010} and references therein. In particular,  the following asymptotic formula was obtained in \cite{Archil2010}
\begin{equation} \label{largeK}
\sig(K) = \sqrt{2}\left[ \sqrt{\log \frac{K}{F} + \log \frac{S Q_T}{C(K)}} - \sqrt{\log \frac{S Q_T}{C(K)}} \right] +
O\left( \left(\log \frac{S Q_T}{C(K)}\right)^{-1/2} \log \log \frac{S Q_T}{C(K)} \right),
\end{equation}
 \noindent as $K \to \infty$. Here, $C(K)$ is the Call option price corresponding to the given $K,T$, and $F = S Q_t/D_T$ is the forward price. One can observe, that in our case $C(K) = Z$. Therefore, as the boundary condition at $K \to \infty$ we can use
 \begin{equation} \label{largeK1}
 \sig(K) = \sqrt{2}\left[ \sqrt{\log \frac{K}{F} + \log \frac{S Q_T}{Z}} - \sqrt{\log \frac{S Q_T}{Z}} \right].
 \end{equation}
 Since $Z \le S Q_T$, the logarithm $\log (S Q_T/Z)$ is always non-negative.

 However, this asymptotics cannot be used along the whole vertical line of the domain $\Omega$ at $K \to \infty$ in Fig.~\ref{domainOm}. For instance, at $Z = S Q_T$ the reminder of the approximation in \eqref{largeK} is large. In this case one can solve numerically  the equation
 \begin{equation} \label{v1}
BS(K,\sig,T) = Z.
\end{equation}

\paragraph{Boundary conditions at $Z = S Q_T$:}
The condition $Z=Q_T S$ implies that $N(d_1(K,T,\sig)) = 1$ and $N(d_2(K,T,\sig)) = 0$. This means that $\sig \to \infty$. To resolve this when constructing a numerical solution, one can move this boundary to $Q_T (S - \Delta S)$, where $\Delta S \ll S$. Then, using a Taylor series expansion of $BS(K,\sig,T)$ around $\sig \to \infty$, we can find the solution of the equation
\begin{equation} \label{u1}
BS(K,\sig,T) = Q_T (S - \Delta S).
\end{equation}
Taking into account first 8 terms in this expansion, we arrive at the following equation
\begin{align} \label{deltaS}
0 &= \sqrt{2 k}  \Big[-6 \left(\sig^2-20\right) \log^4(k)+24 \left(\sig^4-12 \sig^2+240\right) \log^2(k)+\log^6(k) \\
&-48 \left(\sig^6-4 \sig^4+48 \sig^2-960\right)\Big] +24 \sqrt{\pi } \delta  e^{\frac{\sig^2}{8}} \sig^7, \nonumber
\end{align}
\noindent where $k = S Q_T/(K D_T), \ \delta = \Delta S Q_T/(K D_T)$, so $\delta \ll k$. This equation has three roots, and usually one is negative. Among the other two positive roots we need the smallest one, which approximately corresponds to the solution of \eqref{deltaS} with $\delta = 0$. Then this root aligns with the asymptotics of the solution at large $K$ considered in the previous paragraph. Setting $\delta = 0$ in \eqref{deltaS} gives rise to a cubic algebraic equation which can be solved analytically. Alternatively, \eqref{u1} can be solved numerically.

\section{Another type of PDE} \label{otherPDE}

The \eqref{pdeNew} can be further transformed into another type of the PDE. Indeed, the expression in square brackets in \eqref{pdeNew} can be represented as
\begin{equation} \label{sqBr}
A(T,K,Z,\sig) = \left(\frac{\sig_K}{\sig_Z} \right)^2 \sig_{ZZ} - 2 \frac{\sig_ K}{\sig_Z}  \sig_{KZ} + \sig_{KK} = \sig_Z
\left[ \fp{ \left( \dfrac{ \sig_K} {\sig_Z}\right)}{K}  - \frac{1}{2} \fp{\left(\dfrac{\sig_K}{\sig_Z} \right)^2}{Z} \right].
\end{equation}
From \eqref{deltaF}, and \eqref{deltaK} we have
\begin{equation} \label{sigKZ}
\frac{\sig_K}{\sig_Z} = - BS_K(S,\sig,T) = e^{-r T} N(d_2(K,\sig)),
\end{equation}
Substituting \eqref{sigKZ} into \eqref{sqBr}, and taking into account that $\sig = \sig(K,T,Z)$, we obtain
\begin{align} \label{sqBr2}
A(T,K,Z) &= a(T,K,\sig) \sig_Z \Big[ b(T,K,\sig) \sig_Z + c(T,K,\sig)\sig_K + d(T,K,\sig) \Big], \\
a(T,K,\sig) &= \frac{1}{\sqrt{2\pi}} \exp \left(- r T - \frac{1}{2}d_2(K,\sig)^2\right), \nonumber \\
c(T,K,\sig) &= \fp{d_2(K,\sigma,T)}{\sigma}\Big|_{\sigma \sqrt{T} = \sig}, \qquad
b(T,K,\sig) = - e^{-r T} N(d_2(K,\sig)) c(T,K,\sig), \nonumber \\
d(T,K,\sig) &= \fp{d_2(K,\sigma,T)}{K}\Big|_{\sigma \sqrt{T} = \sig}. \nonumber
\end{align}
With this expressions, the PDE in \eqref{pdeNew} takes the form
\begin{align} \label{pdeNew2}
\sig_{T} &= \frac{a(T,K,\sig) K^2}{2 T} \sig^2 \sig_Z \left[b(T,K,\sig) \sig_Z + c(T,K,\sig)\sig_K + d(T,K,\sig)\right] \\
&- (r-q) K \sig_K + q Z \sig_Z + \frac{1}{2 T}\sig. \nonumber
\end{align}
In contrast to \eqref{pdeNew}, this is a nonlinear \textit{hyperbolic} PDE. It can be used as an alternative to \eqref{pdeNew} which is a quasilinear \textit{parabolic} PDE. Accordingly, the PDE in \eqref{pdeNew2} requires just the boundary conditions at lines $K \in [S Q_t/D_T, \infty)$ and $Z = (S Q_T - D_T K)^+$, see Fig.~\ref{domainOm}.

The \eqref{pdeNew2} can be further simplified by using \eqref{sigKZ} that yields
\begin{align} \label{pdeNew3}
\sig_{T} &= -\Theta(T,K,\sigma)\Big|_{\sigma \sqrt{T} = \sig} \sig_Z + \frac{1}{2 T}\sig,
\end{align}
\noindent where $\Theta(T,K,\sigma)$ is the option Theta. This, of course, also follows from \eqref{sigtF}, \eqref{sigzF}, so \eqref{pdeNew3} can be obtained directly from these equations. The \eqref{pdeNew} is a one-dimensional non-linear PDE with $K$ being a dummy variable. In other words, for any fixed $K$ this is a one-dimensional non-linear PDE with $T,Z$ be the independent variables.

\section{Numerical method} \label{solPDE}

In this Section we describe a numerical method used to solve either the PDE in \eqref{pdeNew} or in \eqref{pdeNew2}. The method is constructed similar to how this is done in \cite{Itkin2018}.

First, we re-write \eqref{pdeNew} and \eqref{pdeNew2} in the operator form by introducing an operator $\calL(T,K,Z):  [0,\infty)  \mapsto [0,\infty)$. With this notation \eqref{pdeNew} can be represented as
\begin{equation} \label{pdeNewOp}
\sig_{T} = \calL(T,K,Z) \sig,
\end{equation}
\noindent where
\begin{align} \label{opParab}
\calL(T,K,Z) &= \frac{1}{2 T} \sig^2 K^2 \left[ \left(\frac{\sig_K}{\sig_Z} \right)^2 \pd^2_Z - 2 \frac{\sig_ K}{\sig_Z}  \pd_K \pd_Z  + \pd^2_K \right] - (r-q) K \pd_K + q Z \pd_Z + \frac{1}{2 T},
\end{align}
Similarly, \eqref{pdeNew2} can be re-written in the same form with
\begin{align}
\calL(T,K,Z) &= \frac{a(K,T,Z) K^2}{2 T} \sig^2 \sig_Z \left[b(K,T,Z) \pd_Z + c(K,T,Z)\pd_K + d(K,T,Z)\right] \\
&- (r-q) K \pd_K + q Z \pd_Z + \frac{1}{2 T}. \nonumber
\end{align}

Suppose we now discretize the time to maturity $T \in [0, \infty)$ by creating a uniform grid in time with step $\Delta T$, so our discrete time is defined at the points $[0, \Delta T, 2 \Delta T,\ldots)$. By using Taylor series expansion it can be shown that the discrete solution of the \eqref{pdeNewOp} up to $O((\Delta T))$ could be represented in the form
\begin{equation} \label{Scheme}
\sig(T+\Delta T, K, Z) = e^{\Delta T \mathcal{L}(T,K,Z)} \sig(T, K, Z)
\end{equation}
\noindent where $e^{\Delta T \mathcal{L}(T,K,Z)} $ is the operator exponent, \cite{hochschild1981basic}.

Again, using the Taylor series expansion, one can verify that the following scheme
\begin{equation} \label{Scheme2}
\sig(T+\Delta T, K, Z) = e^{\Delta T \mathcal{L}(T+\Delta T,K,Z)} \sig(T, K, Z)
\end{equation}
\noindent also approximates \eqref{pdeNewOp} up to in $O((\Delta T))$. Combining them together we obtain the scheme
\begin{equation} \label{SchemeF}
\sig(T+\Delta T, K, Z) = e^{\frac{1}{2}\Delta T \left[\mathcal{L}(T+\Delta T,K,Z)+ \mathcal{L}(T,K,Z)\right]} \sig(T, K, Z),
\end{equation}
 \noindent which approximates \eqref{pdeNewOp} up to $O((\Delta T)^2)$.

Now we setup an iterative algorithm to solve this discrete equation. This algorithm is a version of the fix-point Picard iteration algorithm, and relies on the Banach fixed-point theorem, \cite{FPT2003}:
\begin{enumerate}
\item At the first iteration we start by setting $\mathcal{L}(T+\Delta T,K,Z) = \mathcal{L}(T,K,Z)$ as the initial guess.
Thus, \eqref{SchemeF} transforms to \eqref{Scheme}, which could be represented in the form
\begin{equation} \label{k1}
\sig^{(1)}(T+\Delta T, K, Z) = e^{\Delta T \mathcal{L}(T,K,Z)}\sig(T, K, Z),
\end{equation}
\noindent where in $\sig^{(k)}$ the superscript $^{(k)}$ marks the value of $\sig$ found at the $k$-th iteration of the numerical procedure. As the exponent $\mathcal{L}(T,K,Z)$ in the right hands size of \eqref{k1} is computed at the known time $T$, it is a \textit{linear} operator with coefficients being the known functions of $(K,T,Z)$.
In particular, when $T=0$ these functions are determined by the initial condition. If $T>0$ they are determined by the solution at the previous time step.

The \eqref{k1} could be solved by either computing the discrete matrix exponential defined on some grid in $(K,Z)$ (which could be done with complexity $O(N^3 M + N M^3)$, $N$ is the number of the grid points in $K$ direction, and $M$ is the number of the grid points in $Z$ direction), or by using any sort of \pade approximation of the exponential operator. For instance, the \pade approximation $(1,1)$ leads to the well-known Crank-Nicholson scheme which could be solved with the linear complexity in $N$ and $M$, and provides $O((\Delta T)^2)$ approximation in time, see \cite{ItkinBook}.
\newline

\item To proceed we represent \eqref{SchemeF} in the form
\begin{align} \label{kExp}
\sig^{(k)}(T+\Delta T, K, Z) = e^{\frac{1}{2}\Delta T \left[\mathcal{L}^{(k-1)}(T+\Delta T,K,Z)+ \mathcal{L}(T,K,Z)\right]} \sig^{(k-1)}(T, K, Z), \qquad k > 1.
\end{align}
Therefore, the next approximation $\sig^{(k)}(T+\Delta T, K, Z)$ to $\sig(T+\Delta T, K, Z)$ can be found again by either computing the matrix exponential, or by using some \pade approximation.
\end{enumerate}

Having this machinery in hands we can proceed in the same manner until the entire procedure converges, i.e. when the condition
\begin{equation*}
 \| \sig^{(k)}(T+\Delta T, K, Z)  - \sig^{(k-1)}(T+\Delta T, K, Z)\| < \varepsilon
\end{equation*}
\noindent is reached after $k$ iterations with $\varepsilon$ being the method tolerance, and $\|\cdot\|$ being some norm, e.g., $L^2$.

Note, that \eqref{kExp} with the accuracy $O((\dT)^2)$ can be re-written in the form
\begin{align} \label{pdeNewOp1}
\sig_{T}^{(k)} &= \frac{1}{2}\left[\calL(T,K,Z) + \calL^{(k-1)}(T+\Delta T,K,Z)\right]\sig^{(k-1)}  = \tilde\calL^{(k-1)} \sig^{(k-1)},
\end{align}
\noindent where tilde in the new notation $\tilde\calL^{(k-1)}$ means that we take a half of sum of the operators $L$ at times $T$ and $T+\Delta T$, and $^{(k-1)}$ means that at time level $T+\Delta T$ we take the value of the operator at the $(k-1)$-th iteration.

\subsection{Coordinate transformation}

Since our PDEs are defined in the domain $\Omega$ which has a trapezoidal shape, we cannot directly apply a finite-difference (FD) method for solving  \eqref{kExp}, while finite elements or Radial Basis Functions methods can be used here with no problem. However, a simple trick makes application of the FDM possible. The trick utilizes the fact that the right vertical boundary of $\Omega$ lies at $K \to \infty$. Thus, based on the asymptotic formula \eqref{largeK1}, in this limit $\sig \to \infty$ as well, regardless of the value of $Z \in [0,S Q_T/D_T]$. Therefore, without any loss of accuracy, the domain $\Omega$ could be replaced with the domain $\tilde\Omega$ depicted in Fig.~\ref{domainOm1}.
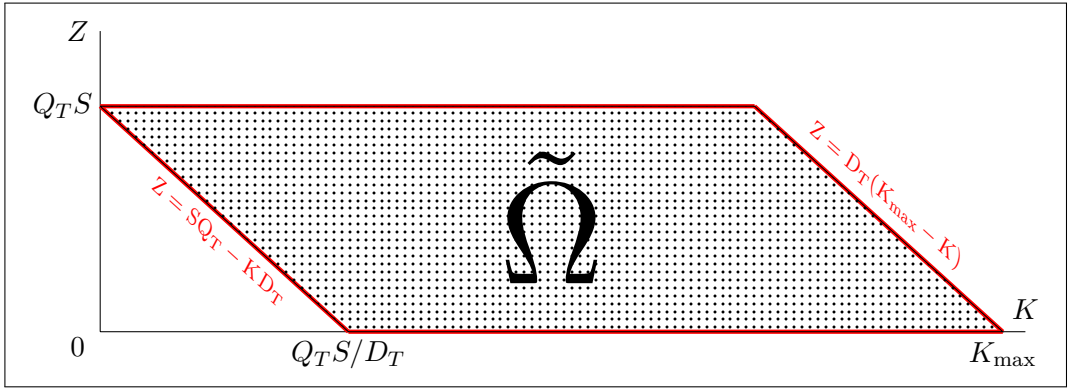
\begin{figure}[!htb]
\captionsetup{format=plain}
\begin{center}
\fbox{
	\begin{tikzpicture}
	\def\sizeV{4.};
	\def\sizeH{12.};
	\def\A{2.5};
	\def\B{3.};
	\def\BB{3.3};
	\def\AB{2.};
	\def\BA{1.5};
	\def\sh{0.03};
	\def\sH{0.03};
	\def\dis{1.5};
	
	\draw (0,0) -- (\sizeH+0.3,0)
	(0,0) -- (0,\sizeV);
	\draw[red, ultra thick]	(0,\B) -- (\sizeH-\BB,\B);
	\draw[red, ultra thick]	(\BB,0) -- (\sizeH,0);	
	\draw[red, ultra thick] (0,\B) -- (\BB,0) node[midway,sloped,below,xslant=-0.5] {\scriptsize $Z = S Q_T - K D_T$};
	\draw[red, ultra thick] (\sizeH-\BB,\B) -- (\sizeH,0) node[midway,sloped,above,xslant=-0.5] {\scriptsize $Z = D_T (K_{\max} - K)$};

	\node at (\sizeH,-0.3){$K_{\max}$};
	\node at (\sizeH+0.3,+0.3){$K$};
	\node at (-0.3,\sizeV){$Z$};
	\node at (-0.3,-0.2){$0$};
	\node at (-0.5,\B){$Q_T S$};
	\node at (\BB,-0.3){$Q_T S/D_T$};
	\node[scale=5] at (\sizeH/2,\B/2){$\tilde\Omega$};
	
	\draw[pattern=dots] (0,\B) -- (\sizeH-\BB,\B) -- (\sizeH,0) -- (\BB,0) -- (0, \B);
	\end{tikzpicture}
}
\caption{The domain $\tilde{\Omega}:  \{K \in [0,\infty), Z \in [0,S Q_T]$.}
\label{domainOm1}
\end{center}
\end{figure}

By construction, the domain $\tilde \Omega$ is a parallelogram, and $K \in [0,K_{\max}]$, where $K_{\max}$ is some fixed strike. We introduce it to truncate an infinite strip in $K$ to some fixed computational area.

Now, the parallelogram could be transformed to rectangle by using, e.g., the following affine transformation (for introduction to this topic see, e.g., \cite{Katsumi1994})
\begin{equation} \label{transform}
Z \mapsto S Q_T u, \qquad K \mapsto \frac{S Q_T}{D_T}(1-u) + \left(K_{\max} - \frac{S Q_T}{D_T}\right)v.
\end{equation}
This transformation converts the domain $\tilde \Omega$ into a unit square with the following map of the boundaries
\begin{align}
Z = S Q_T - K D_T &\mapsto v = 0, \\
Z = D_T (K_{\max} - K) &\mapsto v = 1, \nonumber \\
Z = S Q_T &\mapsto u = 1, \nonumber \\
Z = 0 &\mapsto u = 0. \nonumber
\end{align}
Accordingly, we need to re-write the operator $\calL(T,K,Z)$ in new variables $(u,v)$. For instance, for the parabolic operator \eqref{opParab} this yields
\begin{align} \label{uvOp}
\calL(T,v,u) &= \frac{1}{2 T} \sig^2 K_1^2 \left[
\sig^2_u \pd^2_v
- 2 \sig_u \sig_v  \pd_u \pd_v
+ \sig_v^2 \pd^2_u \right]
- [(r-q) v + B] \pd_v + q u \pd_u + \frac{1}{2 T}, \\
B &= \frac{Q_T S (r (1-u) - q)}{D_T K_{\max} - S Q_T}, \nonumber \\
K_1 &=\frac{c_1 (1 - u) + a v}{a \sig_u + c_1 \sig_v}, \quad a = \frac{1}{S Q_T}, \quad c_1 = \frac{1}{D_T K_{\max} - S Q_T }, \nonumber \\
\end{align}
Since $\tilde \Omega$ - the domain of definition of the operator $\calL(T,v,u)$ in \eqref{uvOp} - is a unit square, now the PDE can be solved by using a FD method.

\subsection{Splitting in spatial dimensions}

Since computation of matrix exponent is time-consuming, especially in case of a two-dimensional operator $L$, we proceed with using a splitting technique, see \cite{LanserVerwer,Duffy,ItkinBook} and references therein. For instance, in \cite{Itkin2017} an ADI (alternative direction implicit) scheme proposed in \cite{HoutWelfert2007} for the solution of a backward PDE was extended twofold. First, a similar ADI scheme has been proposed for the forward equation which equalizes the results obtained by solving backward and forward PDEs. Second, it replaces the first explicit step of the ADI, applied to the mixed derivative term of the PDE, with the fully implicit step. This is especially important for the forward PDE as it increases the stability of the solution and guarantees its positivity.

However, for the PDE in \eqref{pdeNewOp1} using this ADI approach is not a trivial problem, mainly because of various problems with a stable approximation of the mixed derivative term at the second sweep of the ADI. Therefore, instead of the ADI method,
here we use the Strang splitting, \cite{LanserVerwer,Duffy,ItkinBook}. The main idea of this approach is as follows.

With allowance for \eqref{uvOp}, let us represent the \eqref{pdeNewOp1} in the new coordinates in the form
\begin{align} \label{splitting}
\sig_{T}^{(k)} &= \left[ \tilde\calL_u^{(k-1)} + \tilde\calL_v^{(k-1)} + \tilde\calL_{uv}^{(k-1)} \right]\sig^{(k-1)}, \\
\calL_u &= \frac{1}{2 T} \sig^2 K_1^2 \sig_v^2 \pd^2_u  + q u \pd_u + \frac{1}{4 T}, \nonumber \\
\calL_v &= \frac{1}{2 T} \sig^2 K_1^2 \sig_u^2 \pd^2_v - [(r-q) v + B] \pd_v + \frac{1}{4 T}, \nonumber \\
\calL_{uv} &= -\frac{1}{T} \sig^2 K_1^2 \sig_u \sig_v \pd_u \pd_v. \nonumber
\end{align}
Thus, the operator $\calL_u$ includes all differentials with respect to $u$ and also a half of the killing term in \eqref{uvOp}. This is a linear convection-diffusion operator with state and time dependent coefficients. The operator $\calL_v$ includes all differentials with respect to $v$ and a half of the killing term in \eqref{uvOp}, and also is a linear convection-diffusion operator with state and time dependent coefficients. Finally, the operator $\calL_{uv}$ includes only a mixed derivatives term in
\eqref{uvOp}, and also is a linear diffusion operator with state and time dependent coefficients.

Equation \eqref{splitting} can be solved numerically by using the Strang splitting scheme
\begin{align} \label{fwSplit}
\sig^{(k)} (T+\Delta T, u, v) &= e^{\frac{\Delta T}{2} \tilde\calL_u^{(k-1)}}
e^{\frac{\Delta T}{2} \tilde\calL_v^{(k-1)}}
e^{\Delta T \tilde\calL_{uv}^{(k-1)}}
e^{\frac{\Delta T}{2} \tilde\calL_v^{(k-1)}}
e^{\frac{\Delta T}{2} \tilde\calL_u^{(k-1)}}
\sig^{(k-1)} (T+\Delta T, u, v) \\
&+ O((\Delta T)^2). \nonumber
\end{align}
This scheme can also be represented as a sequence of \textit{fractional} steps, which read
\begin{align} \label{fwSplitFin}
\sig^{(k)}_1 &= e^{\frac{\Delta T}{2} \tilde\calL_u^{(k-1)}} \sig^{(k-1)}, \\
\sig^{(k)}_2 &= e^{\frac{\Delta T}{2} \tilde\calL_v^{(k-1)}} \sig^{(k)}_1, \nonumber \\
\sig^{(k)}_3 &= e^{\Delta T \tilde\calL_{uv}^{(k-1)}} \sig^{(k)}_2, \nonumber \\
\sig^{(k)}_4 &= e^{\frac{\Delta T}{2} \tilde\calL_v^{(k-1)}} \sig^{(k)}_3, \nonumber \\
\sig^{(k)} &= e^{\frac{\Delta T}{2} \tilde\calL_u^{(k-1)}} \sig^{(k-1)}_4. \nonumber
\end{align}
We underline that the Strang splitting provides the second order approximation in $\dT$ only if the PDE at every fractional step in \eqref{fwSplitFin} is solved also with the second order approximation in $\Delta T$. To achieve that at every fractional step we use the \pade approximation $(1,1)$ to obtain a Crank-Nicholson scheme in time. For instance, for the first line in \eqref{fwSplitFin} it gives the following scheme
\begin{equation} \label{CN}
\left(1 - \frac{1}{4} \Delta T \tilde\calL_u^{(k-1)} \right)\sig^{(k)}_1 =
\left(1 + \frac{1}{4} \Delta T \tilde\calL_u^{(k-1)} \right)\sig^{(k-1)}.
\end{equation}

\subsection{Spatial discretization}

To apply this algorithm we construct an appropriate discrete nonuniform grid ${\mathbf G}(x)$ in the $(u,v)$ space, see \cite{ItkinBook}. In each spatial direction $i \in [u,v]$ a nonuniform grid contains $N_i+1$ nodes ($p_0, p_1,...,p_{N_i}$) with the steps $h_1 = p_1-p_0, ..., h_{N_i} = p_{N_i} - p_{N_i-1}$. Accordingly, in the operator \eqref{uvOp} we replace continuous derivatives with their finite-difference approximations, and denote these discrete operators defined on ${\mathbf G}(x)$ as $\triangledown$. Thus, we replace $\pd_x$ with $\triangledown_x$, etc.

\subsubsection{Approximation of derivatives}

Once a non-uniform grid is constructed, the first and second derivatives can be approximated on this grid. The corresponding expressions are obtained by using the Taylor series expansions. Here for the reference we provide just the final results, see, e.g., \cite{ItkinBook, HoutFoulon2010}:

Let $f : \mathbb{R} \to \mathbb{R}$ be any given function, let ${x_i}, i \in \mathbb{Z}$ be any given increasing sequence of mesh points, and $h_i = x_i - x_{i-1}, \forall i$. To approximate the first derivatives of $f(x)$ employ the following formulas.
\begin{flalign} \label{approxD}
&\mbox{\textit{Backward scheme:}} \nonumber &\\
& \triangledown^{1B}_{x} = \alpha_{-2} f(x_{i-2}) + \alpha_{-1} f(x_{i-1}) + \alpha_{0} f(x_{i}) = f'(x_i) + O((h_i + h_{i-1})^2), &\\
&\mbox{\textit{Forward scheme:}} \nonumber &\\
& \triangledown^{1F}_{x} = \gamma_{0} f(x_{i}) + \gamma_{1} f(x_{i+1}) + \gamma_{2} f(x_{i+2}) =
f'(x_i) + O((h_{i+1} + h_{i+2})^2). \nonumber &\\
&\mbox{\textit{Central scheme:}} \nonumber &\\
& \triangledown^{1C}_{x} = \beta_{-1} f(x_{i-1}) + \beta_{0} f(x_{i}) + \beta_{1} f(x_{i+1}) = f'(x_i) + O((h_i + h_{i+1})^2). \nonumber
\end{flalign}
Here the super index $^{1X}$ means the first derivatives with the $X$ approximation, $X \in [B,F,C]$ could be backward (B), forward (F) and central (C), and the coefficients $\alpha, \beta, \gamma$ read
\begin{align*}
	\alpha_{-2} &= \dfrac{h_i}{h_{i-1}(h_i + h_{i-1})}, \quad
	\alpha_{-1} = - \dfrac{h_{i-1}+h_i}{h_i h_{i-1}}, \quad
	\alpha_{0} = \dfrac{h_{i-1} + 2 h_i}{h_{i}(h_{i-1} + h_i)},
	\\
	\gamma_{0} &= - \dfrac{2 h_{i+1} + h_{i+1}}{h_{i+2}(h_{i+2} + h_{i+1})}, \quad
	\gamma_{1} = \dfrac{h_{i+2}+h_{i+1}}{h_{i+2}h_{i+1}}, \quad
	\gamma_{2} = - \dfrac{h_{i+1}}{h_{i+2}(h_{i+2} + h_{i+1})}. \nonumber
	\\
	\beta_{-1} &= - \dfrac{h_{i+1}}{h_{i}(h_{i+1} + h_{i})}, \quad
	\beta_{0} = \dfrac{h_{i+1}-h_{i}}{h_{i+1}h_{i}}, \quad
	\beta_{1} = \dfrac{h_{i}}{h_{i+1}(h_{i+1} + h_{i})}. \nonumber
\end{align*}
As follows from \eqref{approxD}, these expressions provide approximation of the first derivatives with the second order $O(h^2)$.

To approximate the second derivatives of $f(x)$ one can employ the following formulas.
\[
\triangledown^{2C}_{x} = \triangledown^{1F}_{x}\triangledown^{1B}_{x} = \triangledown^{1B}_{x}\triangledown^{1F}_{x}, \quad \triangledown^{2B}_{x} = \triangledown^{1B}_{x}\triangledown^{1B}_{x}, \quad \triangledown^{2F}_{x} = \triangledown^{1F}_{x} \triangledown^{1F}_{x}.
\]
These expressions provide approximation of the second derivatives also with the second order $O(h^2)$. For instance, for $\triangledown^{2C}_{x} $ we have in the explicit form
\[ \triangledown^{2C}_{x} = \delta_{-1} f(x_{i-1}) + \delta_{0} f(x_{i}) + \delta_{1} f(x_{i+1}) =
f''(x_i) + O((h_i + h_{i+1})^2), \]
\noindent where
\begin{align*}
	\delta_{-1} &=  \dfrac{2}{h_{i}(h_{i+1} + h_{i})}, \quad
	\delta_{0} = -\dfrac{2}{h_{i+1}h_{i}}, \quad
	\delta_{1} = \dfrac{2}{h_{i+1}(h_{i+1} + h_{i})}. \nonumber
\end{align*}

Let $f : \mathbb{R}^2 \to \mathbb{R}^2$ be any given function of two variables $(x,y)$.
To approximate the mixed derivative $\partial_{x,y} f(x_i, y_j)$ at any point $(x_i, y_j)$ one can consequently apply either one-sided approximations in each direction, or apply the central difference approximation first in $x$ and then in $y$ (or vice versa).
The latter results in the FD approximation of the mixed derivative using a 9-points stencil while still providing the second order of approximation. However, as shown in \cite{Itkin2017}, this discretization could be unstable. Therefore, instead a fully implicit scheme was proposed which, as applied to our problem, is discussed in the next Section.

For the future we need some additional notation. We denote a matrix of $\triangledown^{1F}_x$ as $A^{2F}_{1,x}$; the matrix of  $\triangledown^{1B}_x$ - as $A^{2B}_{1,x}$; the matrix of $\triangledown^{1C}_x$ - as $A^{2C}_{1,x}$
\footnote{Hence, e.g. for $A^{2F}_{1,x}$, the superscript means: backward second order approximation, and subscript means: first order derivative.}. For the first order approximations we use the following definitions. Define a one-sided \textit{forward} discretization of $\triangledown$, which we denote as $A^{1F}_{1,x}: \ A^{1F}_{1,x} C(x) = [C(x+h,t) - C(x,t)]/h$. Also define a one-sided \textit{backward} discretization of $\triangledown$, denoted as $A^{1B}_{1,x}: \ A^{1B}_{1,x} C(x) = [C(x,t) - C(x-h,t)]/h$. These discretizations provide first order approximation in $h$, e.g., $\partial_x C(x) = A^{1F}_{1,x} C(x) +  O(h)$.  Also $I_x$ denotes a unit matrix.

\subsubsection{Approximation of the operator $\calL(T,v,u)$}

Here we use the same approach as in \cite{Itkin2018}. Denote $A^{2C}_{2,x}$ a matrix of the discrete operator $\triangledown^{2C}_{x}$ on the grid ${\mathbf G}(x)$. Observe, that $A^{2C}_{2,x}$ is the Metzler matrix, see \cite{BermanPlemmons1994,ItkinBook}. Indeed, it has all negative elements on the main diagonal, and all nonnegative elements outside of the main diagonal. Also $A^{2C}_{2,x}$ is a tridiagonal matrix.

By the properties of the Metzler matrix $M$ its exponent is a nonnegative matrix. Therefore, operation ${e^M} \sig$ preserves positivity of vector $\sig$. Also all eigenvalues of the Metzler matrix have a negative real part. Therefore, the spectral norm of the matrix follows
\[ \| e^M \| < 1. \]
Thus, if in \eqref{kExp} the operator $\mathcal{M} = \tilde\calL^{(k-1)}$ on the grid ${\mathbf G}(x)$ is represented by the Metzler matrix $A_\calM$, the map $e^{A_\calM}$ is contractual, and, hence, the entire scheme \eqref{kExp} is stable.

By the definition of the operator $\tilde\calL^{(k-1)}$ in \eqref{splitting}, it consists of several terms.

\paragraph{Second derivatives.} The first term
$\frac{1}{2 T} \sig^2 K_1^2  \sig^2_u \pd^2_v$ can be approximated on the grid as
\[A_{2,v} = \frac{1}{2 T} \sig^2 K_1^2  \sig^2_u A^{2C}_{2,v}. \]
As $T > 0$, obviously $A_{2,v}$ is the Metzler matrix. Same is true for the matrix $A_{2,u}$:
\[A_{2,u} = \frac{1}{2 T} \sig^2 K_1^2  \sig^2_v A^{2C}_{2,u}. \]

\paragraph{The mixed derivative.} The mixed derivative term in \eqref{splitting} reads: $ -\frac{1}{T}\sig^2 K_1^2 \sig_u \sig_v  \pd_u \pd_v$. Observe, that from \eqref{transform} we have
\begin{align*}
\sig_u &= \frac{S Q_T}{D_T} \left(D_T \sig_Z - \sig_K \right), \qquad
\sig_v = \sig_K \left(K_{\max} - \frac{S Q_T}{D_T} \right),
\end{align*}
\noindent from \eqref{sigzF}
\[ \sig_Z = \frac{\sqrt{T}}{BS_{\sigma}(K,\Sigma,T) } \ge 0, \]
\noindent and from \eqref{deltaF}, \eqref{deltaK}
\[ \sig_K = - \sig_Z BS_K(S,\Sigma,T) = \sig_Z D_T N(d_2(K,\Sigma,T)) \ge 0. \]
With the allowance for these expressions, it follows that
\begin{align} \label{sigU}
\sig_u &= \frac{\sqrt{T}}{BS_{\sigma}(K,\Sigma,T) } S Q_T \left[1 - N(d_2(K,\Sigma,T)) \right] = s R(d_2)  \ge 0, \\
 \sig_v &= \frac{\sqrt{T}}{BS_{\sigma}(K,\Sigma,T) }  N(d_2(K,\Sigma,T))  (D_T K_{\max} - S Q_T)  =
 \left(\frac{K_{\max}}{K} - s \right) \left(\sqrt{2 \pi} e^{d_2^2/2} -  R(d_2) \right) \ge 0, \nonumber \\
 s &= \frac{S Q_T}{K D_T}, \qquad d_2 = d_2(K,\Sigma,T), \nonumber
\end{align}
\noindent
\noindent where $R(x) \equiv \dfrac{1-N(x)}{N'(x)}$ is the Mills ratio for the standard normal distribution. It
can be efficiently computed by the particularly simple continued fraction representation at $ x > 1$
\begin{equation}
R(x) = \cfrac{1}{x + \cfrac{1}{x + \cfrac{2}{x + ...} } },
\end{equation}
\noindent or by using Taylor series expansion at $ 0 \le x \le 1$, see \cite{MillsRatio2013}.

Thus, $\sig_u \sig_v \ge 0$. Therefore, to have  $A_{2,uv}$ to be the Metzler matrix, we need the matrix of $\sig_u \sig_v$ to be the negative of the Metzler matrix
\footnote{Rigorously speaking, the matrix $A_{2,uv}$ should be the negative of the EM-matrix, where the latter stays for the Eventually M-Matrix. For the introduction into the EM matrices and their properties see \cite{ItkinBook}.},
again see \cite{ItkinBook}. However, the standard methods fail to provide discretization of the second order that guarantees this feature. Therefore, we need another approach, similar to what was proposed in \cite{Itkin2017}. We discuss this in Appendix~\ref{app1}.

As compared with \cite{Itkin2017}, the method in this paper has two innovations. First, in \cite{Itkin2017} a first order approximation in time step was constructed, while here we provide the second order approximation. Second, in \cite{Itkin2017} a mixed derivative term has the form $\rho_{u,v}f(v)g(u) \partial_u \partial_v$, where $f(u)$ is some function of $u$ only, and $g(v)$ is some function of $v$ only. Here we deal with a more general term $\rho_{u,v}f(u,v) \partial_u \partial_v$, and show that the approach of \cite{Itkin2017} can be applied in this case as well.

\paragraph{Convection terms.} The next step is to construct a suitable approximation for the convection terms. The term $q u \partial_u$  admits the \text{forward} approximation $\triangledown^{1F}_{u}$ as $q \ge 0$ and $u \in [0,1]$. Approximating the term $- [(r-q) v + B] \pd_v$ is a more delicate issue and depends on the sign of the function $F(u,v) = (r-q) v + B$. Suppose $r > q$. A typical behavior of function $F(u,v)$ at $T=0.01$ is presented in Fig.~\ref{convecF} at the values of model parameters given in Table~\ref{inputs}.

\begin{figure}[!htb]
\begin{center}
\captionsetup{width=.75\linewidth}
\fbox{\includegraphics[totalheight=2.6in]{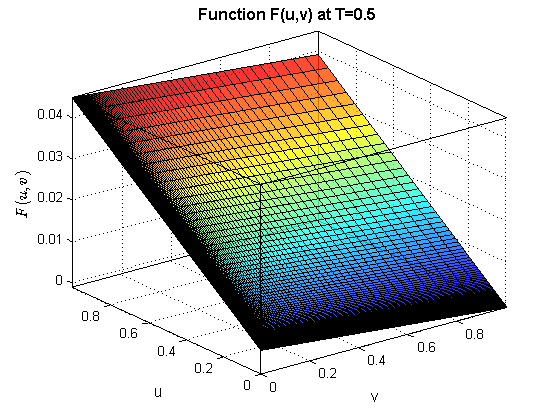}}
\caption{The behavior of function $F(u,v)$ at $T=0.5$.}
\label{convecF}
\end{center}
\end{figure}

As can be seen, $F(u,v) \ge 0$ for all values of $(u,v) \in [0,1]\times[0,1]$. Therefore, when approximating the term $-F(u,v) \pd_v$. a \textit{backward} approximation should be used to obtain the Metzler matrix, i.e. $-F(u,v) A^{1B}_{2,v}$. It can be checked that in case $r < q$ the \textit{forward} approximation should be used as then $F(u,v) \le 0$.

Finally, the resulting matrices $A_{\tilde\calL_u}, A_{\tilde\calL_v},$ are, by construction, also the Metzler matrices, as each of them is a sum of the Metzler matrices, \cite{BermanPlemmons1994}.

It is easy to prove that the matrix $\exp(A_{\tilde\calL^{(k-1)}})$ is the second order approximation of the operator $e^\calM$ in the spatial steps $h_u, h_v$ on the grid, i.e. $O(h_u^2 + h_v^2 + h_u h_v)$. That follows from the fact that the derivatives are approximated by using \eqref{approxD}, where every line is the second order approximation to the corresponding partial derivative.

\subsection{Convergence of Picard iterations}

The remaining point to prove is the convergence of the fixed-point Picard iterations. Again, in doing so we follow \cite{Itkin2018}. According to \eqref{kExp}, change of variables in \eqref{transform}, and \eqref{pdeNewOp1} we run the iterations
\begin{align} \label{kExp1}
\sig^{(k)}(T+\Delta T, u, v) = e^{\dT \tilde\calL^{(k-1)}} \sig^{(k-1)}(T, u, v), \qquad k > 1.
\end{align}

By construction, at every iteration $k$ the operator $\tilde\calL^{(k)}$ is linear, continuous and  bounded at the computational domain. Also, the operator $\calE = \exp[\dT \tilde\calL^{(k)}]$ is Lipschitz
\[
\Big\| \calE \sig_1(T, u, v) - \calE \sig_2(T, u, v) \Big\| \le \Big\| \calE \Big\| \|\sig_1(T, u, v) - \sig_2(T, u, v) \| \le \|\sig_1(T, u, v) - \sig_2(T, u, v)\|, \]
\noindent since  we constructed discretization of $\calE$ as $e^M$ with $M$ being the Metzler matrix, and it was shown earlier that in the spectral norm $\|e^M\| \le 1$.

Now, to prove convergence of the Picard iterations denote by $\hat{\sig}(T+\Delta T,u,v)$ the exact solution of \eqref{pdeNewOp} at time $T+\Delta T$. Then,  by the mean-value theorem for operators, we have
\begin{align*}
\Big\| \sig^{(k)}(T&+\Delta T, u, v) - \hat{\sig}(T+\Delta T, u, v) \Big\|
= \Big\| \calE \sig(T, u, v)  - \hat{\calE}l \sig(T, u, v) \Big\| \\
&\le \| \mathbb{D}(\calE)(\xi^{(k}))\| \cdot  \Big\| \sig^{(k-1)}(T+\Delta T, u, v) - \hat{\sig}(T, u, v) \Big\| \nonumber
\end{align*}
\noindent where $\xi^{(k)}$ is a convex combination of $\sig^{(k-1)}(T+\Delta T, u, v)$ and $\sig^{(k)}(T+\Delta T, u, v)$, and $\mathbb{D}$ denotes the \frechet derivative of the operator $\calE$ at the space of all bounded non-linear operators, see, e.g., \cite{hutson2005applications,LiLuWangMcCammon2010}. Thus, the iterations converge if $\| \mathbb{D}(\calE) \| < 1$. We prove that actually this is the case for the operators in \eqref{splitting} in Appendix~\ref{app2}. In particular, we show there that the iterative scheme is "almost" unconditionally stable, where "almost" means for the values of the model parameters reasonable from the practical point of view.

All the results obtained in this Section can be further summarized in the following proposition
\begin{proposition} \label{Prop1}
The scheme \eqref{kExp} with approximation \eqref{kExp1} is a) (almost) unconditionally stable, b) preserves non-negativity of the solution, c) provides second order approximation of \eqref{kExp} in space on the grid $\mathbb{G}$, and d) converges.
\end{proposition}

\section{Numerical results and comparison} \label{results}

In our numerical test we use the input parameters shown in Table~\ref{inputs}.
\begin{table}[!htb]
\begin{center}
\begin{tabular}{|c|c|c|c|c|c|c|c|c|}
\hline
\rowcolor[rgb]{0,1,1}
$S$ & $r$ & $q$ & $T_{max}$ & $K_{\max}$ & Call/Put & $N_u$ & $N_v$ & $R_c$ \\
\hline
100.00 & 0.05 & 0.01 & 1.00 & 400.00 & Call & 200 & 100 & 20  \\
\hline
\end{tabular}
\caption{Parameters of the test.}
\label{inputs}
\end{center}
\end{table}
We build a non-uniform grid ${\mathbf G}(x)$ in the $(u,v)$ space, which contains $N_u$ nodes in the $u$ direction and $N_v$ nodes in the $v$ direction. Since the  gradients of the implied volatility $\sig(T,u,v)$ are high at the boundaries of the computational domain, i.e. close to $u = 0,1$ and $v=0,1$ it makes sense to make the grid denser close to these areas. Since from practical purposes the far boundary is of a less interest, we compress the grid close to the lines $u=0, \ v=0$. This is obtained by using the compression ratio $R_c=50$, see the description in \cite{ItkinBook}. Thus constructed grid is represented in Fig.~\ref{grid}. With the initial parameters in Table~\ref{inputs}, one step in the $v$ direction approximately corresponds to a change in strike by 3\$, and one step in the $u$ direction approximately corresponds to a 50 cents change in both the options price and the strike.
\begin{figure}[!htb]
\captionsetup{format=plain}
\begin{minipage}{0.4\textwidth}
\begin{center}
\onelinecaptionsfalse
\includegraphics[totalheight=2.6in]{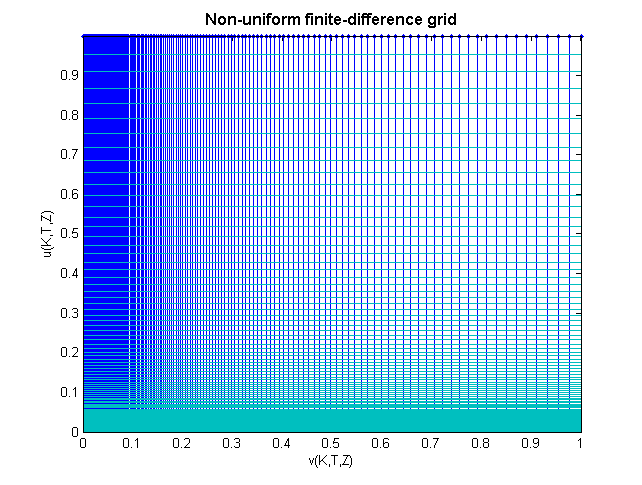}
\caption{Non-uniform finite-difference grid in $(u,v)$ space.}
\label{grid}
\end{center}
\end{minipage}
\hspace{0.1\textwidth}
\begin{minipage}{0.4\textwidth}
\begin{center}
\onelinecaptionsfalse
\includegraphics[totalheight=2.6in]{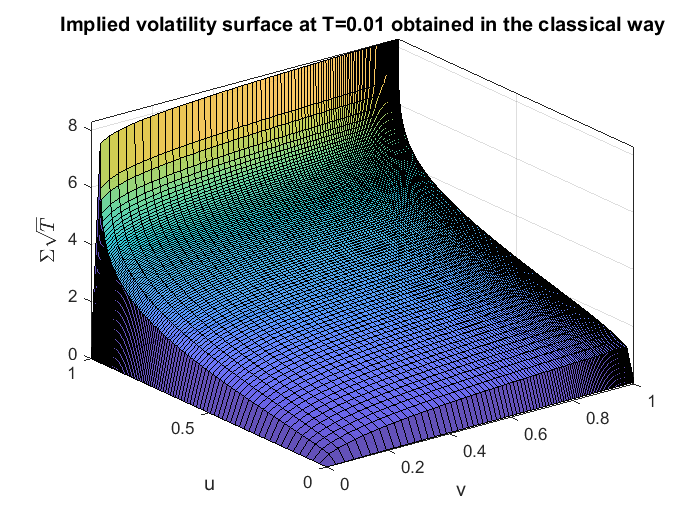}
\caption{Implied volatility $\sig(K,Z,T)$ at $T = 0.01$ computed by using the traditional approach.}
\label{bcFig}
\end{center}
\end{minipage}
\end{figure}

We solve the forward equation starting at $T=0$ and going forward in time to $T_{\max}$ with the time step $\Delta T = 0.01$. The boundary conditions are computed as this is described in Section~\ref{bc} using numerical solutions of \eqref{u0}, \eqref{v1} and \eqref{u1}. Thus found boundary values corresponding to $T = 0.01$ are presented in Fig.~\ref{bcFig} (along the lines $u=0$, $u=1$, $v=0$ and $v=1$).

The standard way to get the implied volatility by solving \eqref{def} numerically is used for comparison, and further is called "traditional".  We run this traditional approach at every point of the grid ${\mathbf G}(x)$ by using Matlab function \url{blsimpv}, or a root solver at the boundaries where \url{blsimpv} provides unsatisfactory results. All tests were run on a PC with Intel i7-4790 CPU (8 Cores, 3.6GHz) as a sequential loop, i.e. no parallel tools have been used. The total elapsed time was 113 secs for 20000 points on the grid, or, on average, 6 msc per a single point.. The results are presented in Fig.~\ref{bc}.

The grid ${\mathbf G}(x)$ doesn't change in time. On the other hand, functions $d_1(\sig), d_2(\sig)$ in the Black-Scholes formula (see \eqref{calldelta}) also almost don't depend on $T$ (since $T$ is now embedded into $\sig$). Therefore, the only dependence on $T$ in the numerical implementation of \eqref{def} is the dependence of $Q_T, D_T$ which is weak. Thus,
Fig.~\ref{bcFig} changes with time very slow. Also, for the purpose of a better illustration of the behavior of $\sig$ at intermediate values of $(u,v)$, we zoom-in Fig.~\ref{bcFig} by excluding the boundaries and setting $T=0.5$. Thus obtained plot is represented in Fig.~\ref{resTradz}.
\begin{figure}[!htb]
\begin{center}
\captionsetup{width=.75\linewidth}
\fbox{\includegraphics[totalheight=2.6in]{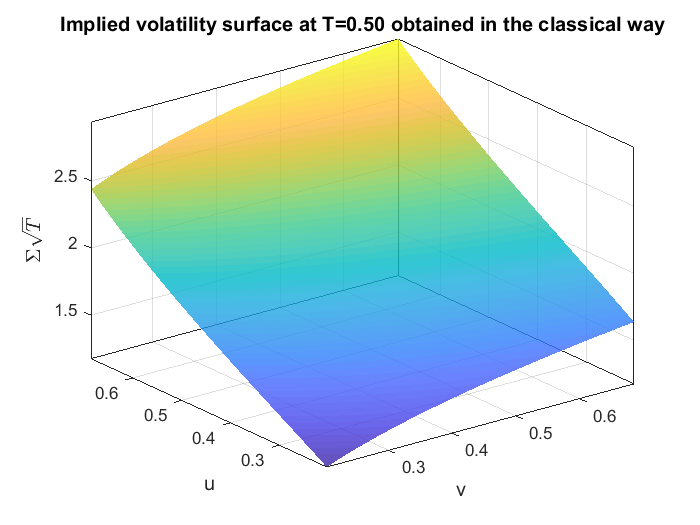}}
\caption{Implied volatility $\sig(K,Z,T)$ at $T = 0.5$ computed by using the traditional approach (zoomed-in).}
\label{resTradz}
\end{center}
\end{figure}

\subsection{First step in time} \label{firstStep}

The first step in time is the most challenging problem with our approach. Indeed, according to Section~\ref{bc}, the initial condition for the PDE in \eqref{pdeNewOp1} is $\sig(0,u,v)  = 0$ for all $u$ and $v$. This, in turn, means that all gradients $\sig_u, \sig_v$ tend to infinity at $T=0$. Also,
\begin{equation} \label{lev1}
\tilde\calL^\nl(0,u,v) = 0,
\end{equation}
\noindent where $\calL^\nl$ denotes the non-linear part of the operator $\calL$. Therefore, the iterative scheme in \eqref{pdeNewOp1} becomes incorrect. At the moment we don't have a good resolution of this problem. Therefore, instead, for $T=\dT$ we proceed by computing $\sig(\Delta T,u,v)$ in the traditional way, and denote thus obtained implied volatility as $\sig_{tr}(\Delta  T,u,v)$.

Again, the main problem here is that at $T=\dT$ we don't have a good initial guess for $\sig$ as it cannot be taken from the previous time step. But, in principle, if we would have a good initial guess for $\sig$, we could run our iterative algorithm even at $T=\dT$. To validate this, we further use $\sig_{tr}(\Delta  T,u,v)$ as the initial guess for the iterative scheme. But, based on \eqref{lev1}, for this time step we also need to modify the scheme in \eqref{pdeNewOp1} to be
\begin{align} \label{pdeNewOp1_0}
\sig_{T}^{(k)} &= \calL^{(k-1)}(T+\Delta T,K,Z) \sig^{(k-1)}  = \tilde\calL^{(k-1)} \sig^{(k-1)}.
\end{align}
Then we solve this equation iteratively, but it is easy to check that this scheme provides only the first order approximation in $\Delta T$.

The results are as follows. We do 20 iterations of the method to obtain the relative error
\[ \varepsilon = \frac{\|\sig_{tr}(\Delta  T,u,v) - \sig(\Delta T,u,v)\|}{\|\sig_{tr}(\Delta  T,u,v)\|} \]
\noindent to be 1.5 bps. The elapsed time for doing so is 16 secs, or 0.8 secs per iteration. The difference between $\sig_{tr}(\Delta  T,u,v)$ and our numerical solution $\sig(\Delta  T,u,v)$ is presented in Fig.~\ref{res1}. The graph is truncated up to $u=v=0.8$ because very large values of strikes don't have practical interest.
\begin{figure}[!htb]
\begin{center}
\captionsetup{width=.75\linewidth}
\fbox{\includegraphics[totalheight=2.6in]{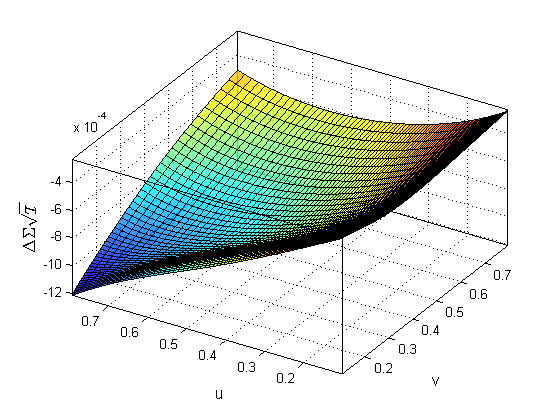}}
\caption{The difference $\Delta \sig = \sig_{tr}(\Delta  T,u,v) - \sig(\Delta T,u,v)$ as a function of $u,v$ at $T = 0.01$ after 20 iterations.}
\label{res1}
\end{center}
\end{figure}

\subsection{Next steps in time}

Having good (not all zeros) values of $\sig(\Delta T,u,v)$ we can proceed based on the general scheme described in Section~\ref{solPDE}. We increase time to $T=2\Delta T$, iterate to get the numerical solution, then compute $\sig_{tr}(2\Delta T,u,v)$, compare these solutions, etc. In Fig.~\ref{res2} these results are presented for $T = 2\Delta T, 5\Delta T, 10\Delta T, 20\Delta T$.
\begin{figure}[!htb]
\begin{center}
\captionsetup{width=.95\linewidth}
\fbox{\includegraphics[width=.45\textwidth]{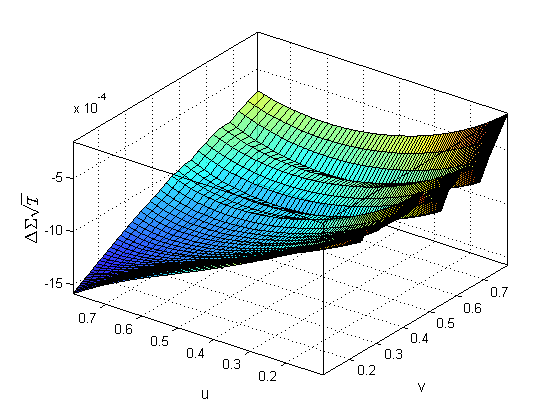}} \quad
\fbox{\includegraphics[width=.45\textwidth]{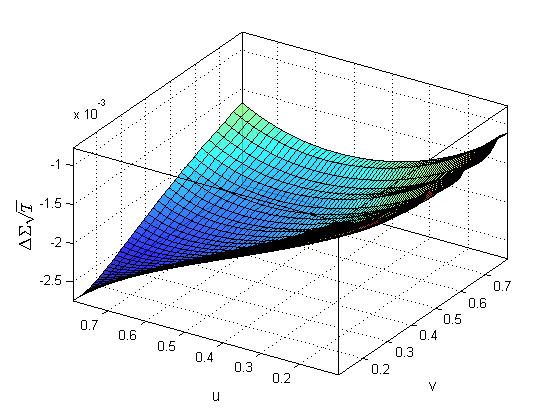}} \\
\vspace{0.1in}
\fbox{\includegraphics[width=.45\textwidth]{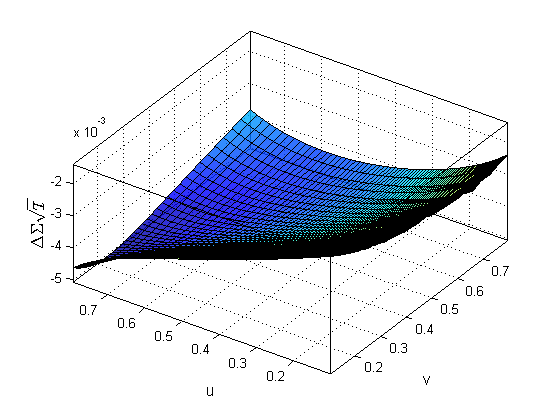}} \quad
\fbox{\includegraphics[width=.45\textwidth]{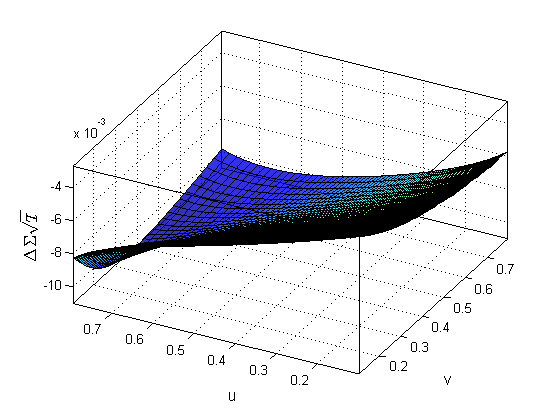}} \\
\caption{The difference $\Delta \sig = \sig_{tr}(\Delta  T,u,v) - \sig(\Delta T,u,v)$ as a function of $u,v$ at $T = 0.02$ (upper-left), $T = 0.05$ (upper-right), $T = 0.1$ (bottom-left) and $T = 0.2$ (bottom-right).}
\label{res2}
\end{center}
\end{figure}

It turns out that at those time steps the iterations converge much faster as compared with the test for the first step reported in Section~\ref{firstStep}. This convergence is summarized in Table~\ref{conv}. This partially could be explained by the order of approximation in time, since at the first step the scheme was of the order one, while at the next steps it is of the order 2.
\begin{table}[!htb]
\begin{center}
\begin{tabular}{|r|r|r|r|}
\hline
\rowcolor[rgb]{0,1,1}
Step in time & Iterations per step & $\varepsilon \ \ $ & $t$ elapsed, secs \\
\hline
1 & 20 & 1.5e-4 & 16 \\
\hline
2 & 3 & 2.3e-6 & 2.7 \\
\hline
3 & 3 & 1.9e-6 & 2.8 \\
\hline
5 & 3 & 1.4e-6 & 2.8 \\
\hline
10 & 3 & 8.3e-7 & 2.8 \\
\hline
20 & 3 & 3.8e-7 & 2.8 \\
\hline
\end{tabular}
\caption{Convergence of the Picard iterations for various steps in time $T$. The time step $\Delta T$ = 0.01.}
\label{conv}
\end{center}
\end{table}

As this can be seen from the presented results, the convergence of the method slightly improves with time $T$, the elapsed time remains the same, while the relative error between the traditional method and the PDE solver slightly increases with $T$. The later can obviously be explained by accumulation of the approximation errors at every time step with the increase of $T$. However, for instance, for $T=0.2$ this error is still acceptable (to be 20-50 bps). On the other hand, the biggest deviation from the traditional results corresponds to high $u$ and $v$. In particular, in our numerical example the point $u=v=0.8$ maps to the strike $K = 260\$ $ which, in turn, corresponds to the moneyness $S/M = 0.385$ and the option price 20\$. Hence, this is a deep OTM option, and the corresponding implied volatility is very high, 2.61. Therefore, this case is almost impractical.

\section{Discussion} \label{discussion}

In this paper we derive a backward and forward quasilinear parabolic PDEs that govern the implied volatility of a contingent claim whenever the latter is well-defined. This would include at least any contingent claim written on a positive stock price whose payoff at a possibly random time is convex. Alternative to that we also derived a forward nonlinear hyperbolic PDE of the first order which also governs evolution of the implied volatility in $(K,T,Z)$ space.
We also discuss suitable initial and boundary conditions for those PDEs. Finally, we demonstrate how to solve them numerically by using an iterative method. All these results are new.

We compare performance of two methods. The first one uses a traditional scheme of getting the implied volatilities by using a root solver and solving \eqref{def}. The second method relies on a numerical solution of the derived PDEs. We show that based on our numerical experiments, if the implied volatility is needed at every point on the rather dense grid, the PDE solver significantly outperforms the root solver. As follows from Table~\ref{conv}, the PDE solver is almost 40 times faster.

However, a few comments shoud be made with this regard. First, our method utilizes the traditional method in two ways. We use it to compute the boundary conditions, and we also use it at the first time step $T = \Delta T$ as this was discussed in Section~\ref{firstStep} since we don't have a good start value of $\sig$ for iterations to kick off. Also this drops the order of approximation in time from $O((\Delta T)^2)$ to $O(\Delta T)$ at this step. Nevertheless, all those issues are relatively minor, and for the longer time horizons our method still outperforms the traditional one.

The second obstacle is that, most likely, in practice we rarely need implied volatilities on so dense grid in both space and time. Usually, we need them only at those points in $(K,T)$ which correspond to the tradable market strikes and maturities, and they are not so dense. However, the FD method still requires a grid. Therefore, the number of points for the traditional method could be much smaller that that for the FD method. This could compensate a slower performance of the traditional method and make both methods to be comparable in the elapsed time.

Third, as was already mentioned, functions $d_1(\sig), d_2(\sig)$ in the Black-Scholes formula (see \eqref{calldelta}) weakly  depend on $T$ (since $T$ is now embedded into $\sig$), namely only via $Q_T$ and $D_T$. Therefore, the solution of \eqref{def} also depends on $T$ only via $Q_T, D_T$. Therefore, since we already computed $\sig(\Delta T,K,Z)$ by using the traditional way, then for $T=2 \Delta T$ \eqref{def} could be solved by expanding the left hands side into series on $\Delta T \ll 1$. In more detail, we write
\begin{align} \label{def2}
BS(T,K,\sig(T,K,Z)) &= Z, \\
BS(T+\Delta T,K,\sig(T+\Delta T,K,Z)) &= Z, \nonumber
\end{align}
\noindent and define $\sig(T+\Delta T,K,Z) = \sig(T,K,Z) + \Delta \sig(T,K,Z)$. Now expanding the second line of \eqref{def2} into series on $\Delta T \ll 1, \Delta \sig \ll 1$, and using the first line in \eqref{def2}, in the first order we obtain
\begin{equation} \label{dSig}
\Delta \sig(T,K,Z) = - \frac{\Theta(T,K,\sig(T,K,Z))}{\mathrm{Vega}(T,K,\sig(T,K,Z))} \sqrt{T} \Delta T.
\end{equation}
Since $\sig(T,K,Z)$ is already known, the expression in \eqref{dSig} can be easily computed.

Once all $\sig(T+\Delta T,K,Z)$ are computed in such a way, we can proceed to $T+2\Delta T$. The problem with this approach, however, is that its accuracy decreases with every time step. As an alternative, periodically at some time steps the accurate value of $\sig(T,K,Z)$ can be recomputed by using the traditional approach. However, same could be done for the PDE approach as well.

It should also be mentioned that in this paper we present results for European plain vanilla options using the forward equation, while for e.g., for American options one has to solve the backward PDE, which computationally is more expensive.

\section*{Acknowledgments}
We thank Archil Gulisashvili and Daniel Duffy for useful discussions. Any errors are our own.

\vspace{0.3in}


\newcommand{\noopsort}[1]{} \newcommand{\printfirst}[2]{#1}
  \newcommand{\singleletter}[1]{#1} \newcommand{\switchargs}[2]{#2#1}

\appendix
\numberwithin{equation}{section}
\appendixpage
\section{FD approximation of the mixed derivative term in \eqref{splitting}} \label{app1}

In this Appendix we show how to solve the third line in \eqref{fwSplitFin} with the accuracy $O((\dT)^2)$ in time and
$O(h_u^2 + h_v^2 + h_u h_v)$ in space folloiwng the method first proposed in \cite{Itkin2017}. To recall, the equation under consideration reads
\begin{align} \label{mixEq}
\sig^{(k)} &= e^{\Delta T \tilde\calL_{uv}^{(k-1)}} \sig^{(k)}, \\
\calL_{uv} &= -\frac{1}{T}\sig^2 K_1^2 \sig_u \sig_v  \triangledown_u \triangledown_v. \nonumber
\end{align}
To achieve this accuracy in time we use the \pade approximation $(1,1)$ to obtain a Crank-Nicholson scheme (wrtten with some loose of notation to make it better readable)
\begin{equation} \label{CN1}
\left(1 - \frac{1}{4} \Delta T \tilde\calL_{uv}^{(k-1)} \right)\sig(T+\dT) =
\left(1 + \frac{1}{4} \Delta T \tilde\calL_{uv}^{(k-1)} \right)\sig(T) = \bar\sig(T).
\end{equation}
This numerical scheme is, of course, implicit due to the fact that, despite we know $\bar\sig(T)$, to find $\sig(T+\dT)$ we need to solve \eqref{CN1}. It is known that, after some discretization on the FD grid is applied to \eqref{CN1}, the matrix in the left hands side of \eqref{CN1} must be an M-matrix, to ensure that this scheme is stable and preserves the positivity of the solution. However, this is impossible by using the standard FD discretizations of the second order, in more detail see \cite{ItkinBook}.

Therefore, we rewrite \eqref{CN1} by dropping the dependence on $T$ and using a trick proposed in \cite{Itkin2017}\footnote{The trick is motivated by the desire to build an ADI scheme which consists of two one-dimensional steps, because for the 1D equations we know how to make the rhs matrix to be an EM-matrix, see \cite{Itkin2,ItkinBook} and references therein.}
\begin{align} \label{trickNeg}
\Big(P &- \sdt \rho_{u,v} U(u,v) \triangledown_u \Big)\Big(Q + \sdt V(u,v) \triangledown_v \Big) \sig  \\
	&= \bar\sig + \left[(P Q-1) - Q\sdt \rho_{u,v} U(u,v) \triangledown_u + P \sdt V(u,v)  \triangledown_v  \right]\sig. \nonumber \\
U(u,v) &= \frac{1}{\sqrt{T}}\sig K_1 \sig_u, \qquad V(u,v) = \frac{1}{\sqrt{T}}\sig K_1 \sig_v. \nonumber
\end{align}
Here $\rho_{u,v} = -1/4,$\footnote{We introduce this notation to make it easy to compare the description of the method in this paper with that in \cite{Itkin2017}.} $P,Q,$ are some positive numbers which have to be chosen based on some conditions, e.g., to provide diagonal dominance of the matrices in the parentheses in the lhs of \eqref{trickNeg}, see below. Also the
intuition behind this representation is discussed in detail in \cite{Itkin2017}.

Note, that in \cite{Itkin2017} we considered a mixed derivative term of the form $\rho_{u,v}f(v)g(u) \partial_u \partial_v$, where $f(u)$ is some function of $u$ only, and $g(v)$ is some function of $v$ only. Here we are dealing with a more general term $\rho_{u,v}f(u,v) \partial_u \partial_v$, so the presented approach is a generalization of that in \cite{Itkin2017}. Also our method in this paper is of the second order in time while the method in \cite{Itkin2017} is of the first order.

The \eqref{trickNeg} can be solved using fixed-point Picard iterations. One can start with setting $\sig^0 = \bar\sig$ in the right hands side of \eqref{trickNeg}, then solve sequentially two systems of equations
\begin{align} \label{spl2}
\Big(Q &+ \sdt V(u,v) \triangledown_v\Big) \sig^*
	= \bar\sig + \Big[P Q -1  - Q\sdt \rho_{u,v} U(u,v) \triangledown_u
	+ P \sdt V(u,v) \triangledown_v  \Big]\sig^j, \nonumber \\
\Big(P &- \sdt \rho_{u,v} U(u,v) \triangledown_u \Big)\sig^{j+1} = \sig^*.
\end{align}
Here $\sig^j$ is the value of $\sig$ at the $j$-th iteration\footnote{The reader shouldn't miss these iterations (marked by $j$) used to solve \eqref{mixEq} with those (marked by $k$) for the solution of the nonlinear equation which are discussed in Section~\ref{solPDE}.}.


To solve \eqref{spl2} we propose two FD schemes. The first one (Scheme~A) is introduced by the following Propositions\footnote{For the sake of clearness we formulate this Proposition for the uniform grid, but it is transparent how to extend it for the non-uniform grid.}:

\begin{proposition} \label{propPos}
Let us approximate the left hands side of \eqref{spl2} using the following finite-difference scheme:
\begin{align} \label{fd1}
\Big(Q I_v + \sdt V(u,v) A^{2B}_{1,v}\Big) \sig^* 	&= \alpha^{+}\bar\sig  - \sig^j, \\
\Big(P I_u - \sdt \rho_{u,v} U(u,v) A^{2B}_{1,u} \Big)\sig^{j+1} &= \sig^*, \nonumber	\\
\alpha^{+} = (PQ + 1)I &- Q\sdt \rho_{u,v} U(u,v) A^{1F}_{1,u} + P\sdt V(u,v) A^{1F}_{1,v}. \nonumber
\end{align}
Then this scheme is unconditionally stable in time step $\dT$, approximates \eqref{spl2} with $O(\sqrt{\dT}\max(h_u,h_v))$ and preserves positivity of the matrix $\sig(u,v,T)$ if $Q = \beta \sdt/h_v, \ P = \beta \sdt/h_u$, where $h_v, h_u$ are the grid space steps correspondingly in $v$ and $u$ directions, and the coefficient $\beta$ must be chosen to obey the condition:
\[
\beta > \max_{u,v}[V(u,v) + \rho_{u,v} U(u,v)].
\]
\end{proposition}
\begin{proof}
The proof can be found in \cite{Itkin2017,ItkinBook}.
\end{proof}
The computational scheme in \eqref{fd1} should be understood in the following way. At the first line of \eqref{fd1} we begin with computing the product $\sig_1 = \alpha^+ \bar\sig$. This can be done in three steps. First, the product $\sig_1 = Q\sdt \rho_{u,v} U(u,v) A^{1F}_{1,u} \bar\sig$ is computed in a loop on $v_i, i=1,...,N_v$. In other words, if $\bar\sig$ is a $N_u \times N_v$ matrix where the rows represent the $u$ coordinate and the columns represent the $v$ coordinate, each $j$-th column of $\sig_1$ is a product of matrix $Q\sdt \rho_{u,v} U(u,v) A^{1F}_{1,u}$ and the $j$-th column of $\bar\sig$. The second step is to compute the product $\sig_2 = P\sdt V(u,v) A^{1F}_{1,v} \bar\sig$ which can be done in a loop on $u_i, \ i=1,...,N_u$. Finally, the right hand side of the first line in \eqref{fd1} is $(PQ+1) \bar\sig - \sig_1 + \sig_2 - \sig^j$.
Then in a loop on $u_i, \ i=1,...,N_u$, $N_u$ systems of linear equations have to be solved, each giving a row vector of $\sig^*$. The advantage of the representation \eqref{fd1} is that the product $\alpha^+ \bar\sig$ can be precomputed.

If $\sdt \approx \max(h_u,h_v)$, then the whole scheme becomes of the second order in space. However, this would be a serious restriction inherent to the explicit schemes. Therefore, we don't rely on it. Note, that in practice the time step is usually chosen such that $\sdt \ll 1$, and hence the whole scheme is expected to be closer to the second, rather than to the first order in $h$. Note, that this approach  is similar to \cite{toivanen2010,chiarella2008} for negative correlations, where a seven point stencil breaks a rigorous second order of approximation in space.

As shown in \cite{Itkin2017}, the rate of convergence of the Picard iterations is linear.

The above results can be further improved by making the whole scheme to be of the second order of approximation in $h_u$ and $h_v$. We call this FD scheme as Scheme~B.

\begin{proposition} \label{propPos2}
Let us approximate the lhs of \eqref{spl2} using the following finite-difference scheme:
\begin{align} \label{fd13}
\Big(Q I_v + \sdt V(u,v) A^{2B}_{1,v}\Big) \sig^* 	&= \alpha^+_2 \bar\sig - \sig^j, \\
\Big(P I_u - \sdt \rho_{u,v} U(u,v) A^{2B}_{1,u} \Big)\sig^{j+1} &= \sig^*, \nonumber	\\
\alpha^{+}_2 = (PQ + 1)I &- Q\sdt \rho_{u,v} U(u,v) A^{2F}_{1,u} + P\sdt V(u,v) A^{2F}_{1,v}. \nonumber
\end{align}
Then this scheme is unconditionally stable in time step $\dT$, approximates \eqref{spl2} with $O(\max(h^2_u,h^2_v))$ and preserves positivity of the vector $\sig(u,x,T)$ if $Q = \beta \sdt/h_v, \ P = \beta \sdt/h_u$, where $h_v, h_u$ are the grid space steps correspondingly in $v$ and $u$ directions, and the coefficient $\beta$ must be chosen to obey the condition:
\[
\beta > \dfrac{3}{2}\max_{u,v}[V(u,v) + \rho_{u,v} U(u,v)].
\]
The scheme \eqref{fd13} has a linear complexity in each direction.
\end{proposition}
\begin{proof}
See \cite{Itkin2017}.
\end{proof}

The coefficient $\beta$ should be chosen experimentally. In our experiments described in the following sections we used
\begin{equation} \label{beta}
\beta = \max_{u,v}[U(u,v) - \rho_{u,v} U(u,v)].
\end{equation}

\section{Norm of the \frechet derivative for operators in \eqref{splitting}}  \label{app2}

To compute the norm of the \frechet derivative recall that by definition
\[ \| \mathbb{D}(\calE) \| = \sup_{\nu \ne 0} \dfrac{\| \mathbb{D}(\calE)(\nu) \|}
{\| \nu \|}\]
If $\nu = (\nu_1,...,\nu_m) \in [L^\infty(-\infty,0)]^m$, then
\[ \mathbb{D}(\calE)(\nu) = \fp{\calE(\sig + k\nu)}{k}\Big|_{k=0}. \]
In other words, in this case the \frechet derivative coincides with the \gato derivative.

To find $\mathbb{D}(\calE^\nl)(\sig)$ we need some results from functional analysis, namely:
\begin{enumerate}
\item The product and quotient rules for the \frechet derivatives, \cite{schwartz1969nonlinear}.
\item The property that the \frechet derivative of a linear operator is the operator itself, \cite{schwartz1969nonlinear}.
\item Let $A \to A^r$ be the map that takes a positive definite matrix to its $r$-th power, and let $\mathbb{D}A^r$ be the \frechet derivative of this map. Then it is known that $\| \mathbb{D}A^r \| = \| r A^{r-1} \|$ if $r \le 1$ or $r \ge 2$, \cite{Bhatia2003}.
\item Triangle inequality.
\end{enumerate}

Let us denote the non-linear part of $\tilde\calL$ as $\tilde\calL^\nl$, and the linear part - as $\tilde\calL^L$.
Accordingly, $\tilde\calL = \tilde\calL^\nl + \tilde\calL^L$. Using the properties of the \frechet derivatives, we obtain
\begin{align} \label{normU}
\left\| \mathbb{D}(\dT L^\nl_u)(\sig) \right\| &\le \frac{\dT}{2 T} \left\| \left(\frac{\sig_v [c_1 (1 - u) + a v]}{a \sig_u + c_1 \sig_v}\right)^2 (2\sig + \sig^2) \triangledown_u^2
\right\|  \\
&\le \frac{\dT}{2T} \left\| \frac{\sig_v [c_1 (1 - u) + a v]}{a \sig_u + c_1 \sig_v}\right\|^2 \|\sig(2 + \sig)\|
\| \triangledown_u^2 \|. \nonumber
\end{align}
It can be seen that the first norm $\|\cdot\|_1$ in the right hands part of \eqref{normU} is bounded, and moreover, $\|\cdot\|_1 \le 2$. The second norm $\|\cdot\|_2$ is also bounded since we work on a truncated grid with $K < K_{\max}$. Therefore, the values of $\sig$ are also bounded. For instance, in our numerical experiments $\sig < 8$. For the third norm $\| \triangledown_u^2 \|$ it is well-known, \cite{MMCP89}, that at the uniform grid
\begin{equation} \label{normTr}
\| \triangledown_u^2 \| = \max_i \left[ \dfrac{4}{h_u^2}\sin^2 \dfrac{i \pi}{2(N_u+1)}\right], \quad i \in [1,...,N_u],
\end{equation}
\noindent where $N_u$ is the size of the matrix $A_{2,u}^C$, and $h_u$ is the grid step. For the uniform grid
$h_u = 1/(N_u-1)$.
Therefore,
\begin{equation} \label{stabU}
\left\| \mathbb{D}(\dT L^\nl_u)(\sig) \right\| \le 2 \dT \frac{2}{h_u^2}\|\Sigma\|^2
\end{equation}
\noindent where $\|\Sigma\|$ is the norm of matrix $A(\Sigma)$ on the grid.

A similar analysis can be provided for the norm $\left\| \mathbb{D}(\dT L^\nl_v)(\sig) \right\|$ to obtain
\begin{equation} \label{stabV}
\left\| \mathbb{D}(\dT L^\nl_u)(\sig) \right\| \le 2 \dT \frac{1}{h_v^2}\|\Sigma\|^2
\end{equation}
\noindent since
\[ \left\| \frac{\sig_u [c_1 (1 - u) + a v]}{a \sig_u + c_1 \sig_v}\right\|^2 \le 1. \]

Finally for the mixed derivative term we have
\begin{align} \label{normUV}
\left\| \mathbb{D}(-\dT L^\nl_{uv})(\sig) \right\| &\le  \frac{\dT}{T} \left\| \left(\frac{ c_1 (1 - u) + a v}{a \sig_u + c_1 \sig_v}  \right)^2 \sig_u \sig_v  \triangledown_{uv} (2 \sig + \sig^2) \right\|.
\end{align}
Similarly to \eqref{normTr},
\begin{equation} \label{normTrM}
\| \triangledown_u \triangledown_v\| = \max_{i,j} \left[ \dfrac{4}{h_u h_v}\sin \dfrac{i \pi}{2(N_u+1)}
\sin \dfrac{j \pi}{2(N_v+1)} \right], \quad i \in [1,...,N_u], \ j \in [1,...,N_v],
\end{equation}
\noindent and also
\[ \left\| \left(\frac{ c_1 (1 - u) + a v}{a \sig_u + c_1 \sig_v}  \right)^2 \sig_u \sig_v \right\| \le \frac{3}{2}. \]
Therefore,
\begin{equation} \label{stabUV}
\left\| \mathbb{D}(\dT L^\nl_{uv})(\sig) \right\| \le 2 \dT \frac{3}{h_v h_u}\|\Sigma\|^2.
\end{equation}

Finally, using these expressions, and properties of the \frechet derivatives, we obtain
\begin{align} \label{finFre}
\| \mathbb{D}(\calE)(\sig) \| &= \dT \| \calE \| \| \mathbb{D} \left(\fp{\calE}{\dT} \right)(\sig) \|
\le \dT \|\calE \|  \left( \|\Sigma\|^2 \left|\frac{1}{h_u^2} + \frac{2}{h_u^2} -
\frac{3}{h_v h_u} \right| + \| L^L \| \right),
\end{align}
\noindent where $\| L^L \|$ is the norm of the linear part of the operator $\tilde\calL_u^{(k-1)} + \tilde\calL_v^{(k-1)} + \tilde\calL_{uv}^{(k-1)}$ in \eqref{splitting}. It is easy to show that
\[
\| L^L \| \le \alpha^L \left|\frac{1}{h_u} + \frac{1}{h_v}\right|,
\]
\noindent where $|\alpha^L| \le 1$.

As we mentioned already, on the FD grid $\|\calE \|$ is a spectral norm of the matrix exponential $e^\calM$ where
\[ \| \calM \| \le \dT \left( \alpha^\nl \|\Sigma\|^2 \left|\frac{1}{h_u^2} + \frac{2}{h_u^2} -
\frac{3}{h_v h_u} \right| + \alpha^L \left|\frac{1}{h_u} + \frac{1}{h_v}\right| \right).
\]
Here $\alpha^\nl, \alpha^L$ are some terms that don't depend on $\dT, h_u, h_v$, i.e. on the FD steps in all spatial and temporal directions. Also, by construction $M = A(\calM)$ is the Metzler matrix with all eigenvalues being negative. Therefore, it can be seen that by decreasing the steps $h_u, h_v$, or increasing the step $\dT$ we decrease the norm $\| \mathbb{D}(\calE)(\sig) \|$ (as the exponential term dominates). Therefore, to decrease the norm $\| \mathbb{D}(\calE)(\sig) \|$ we can, e.g., fix the time step $\dT$, and then decrease the step $h_u$ or  $h_v$ (or both) until we get $\| \mathbb{D}(\calE)(\sig) \| < 1$.

This could be compared with the familiar stability condition for the Euler explicit FD scheme, which reads, \cite{ItkinBook}
\begin{equation} \label{stabCond}
2 D \frac{\Delta t}{h^2} \le 1,
\end{equation}
\noindent with $D$ being some diffusion coefficient. The latter condition is restrictive, because it sets the upper limit on the time step given the space step $h$. In contrast, our condition, using a loose notation, reads
\[  2 D \frac{\Delta T}{h^2} >  1. \]
Since this is a stability condition, the convergence of our iterative method is conditional. However, our condition
is far less restrictive then that in \eqref{stabCond}. Indeed, given the space step $h$ we need the time step to \textit{exceed} the value $h^2/(2D)$, while \eqref{stabCond} requires it to be \textit{less} than $h^2/(2D)$. Since $h^2$ is normally small, and $D$ is proportional to $\| \Sigma \|^2 > 1$ on the grid, it could be satisfied with no serious restrictions on $\dT$. Therefore, the iterative scheme is "almost" unconditionally stable, where "almost" means - at the values of the model parameters reasonable from the practical point of view.

\end{document}